%% file: IEEE_TCOM.tex
\documentclass[draftcls,onecolumn,12pt]{IEEEtran}
\normalsize

\ifCLASSINFOpdf
\else
\fi
\DeclareUnicodeCharacter{2061}{}
\usepackage{tikz}
\usetikzlibrary{calc}
\usepackage{caption}
\usepackage{subcaption}
\usepackage{amsthm}
\usepackage{cite}
\usepackage{amsmath,amssymb,amsfonts}
\usepackage{mathtools}
\usepackage{algorithmic}
\usepackage{graphicx}
\usepackage{textcomp}
\usepackage{xcolor}
\usepackage{bbold}
\usepackage{mathrsfs}
\def\BibTeX{{\rm B\kern-.05em{\sc i\kern-.025em b}\kern-.08em
    T\kern-.1667em\lower.7ex\hbox{E}\kern-.125emX}}
   \newtheorem{lemma}{Lemma}
   \newtheorem{theorem}{Theorem}
   \newtheorem{Corollary}{Corollary}
   \newtheorem{definition}{Definition}

\usepackage{titlesec}

\titlespacing{\subsection}{0pt}{0.5em}{0pt}

\begin{document}
\title{Finite Blocklength Regime Performance {of} Downlink Large Scale Networks }

\author{
\IEEEauthorblockN{Nourhan Hesham, \textit{Student Member, IEEE}, Anas Chaaban, \textit{Senior Member, IEEE}, Hesham ElSawy, \textit{Senior Member, IEEE}, Jahangir Hossain, \textit{Senior Member, IEEE}}\\
\IEEEauthorblockA{}%School of Engineering, University of British Columbia, Kelowna, BC V1V1V7, Canada\\
%Email: \{nourhan.soliman,anas.chaaban\}@ubc.ca
%}
\thanks{
 N. Hesham, A. Chaaban, and J. Hossain are with the School of Engineering, University of British Columbia, Kelowna, BC V1V 1V7, Canada (e-mail: \{nourhan.soliman,anas.chaaban,jahangir.hossain\}@ubc.ca), and N. Hesham is on leave from the Department of Electronics and Electrical Engineering, Cairo University, Cairo, Egypt.

H. ElSawy is with the School of Computing, Queen’s University, Kingston, ON K7L 2N8, Canada. (e-mail: hesham.elsawy@queensu.ca).

This publication is based upon work supported by King Abdullah University of Science and Technology (KAUST) under Award No. OSR-2018-CRG7-3734.

Part of this work was presented in the IEEE International Conference on Communications (ICC 2021)~\cite{mypaper}
}}

\maketitle

\begin{abstract} 
 Some emerging 5G and beyond use-cases impose stringent latency constraints, which necessitates a paradigm shift towards finite blocklength performance analysis. In contrast to Shannon capacity-achieving codes, the codeword length in the finite blocklength regime (FBR) is a critical design parameter that imposes an intricate tradeoff between delay, reliability, and information coding rate. In this context, this paper presents a novel mathematical analysis to characterize the performance of large-scale downlink networks using short codewords. Theoretical achievable rates, outage probability, and reliability expressions are derived using the finite blocklength coding theory in conjunction with stochastic geometry, and compared to the performance in the asymptotic regime (AR). Achievable rates under practical modulation schemes as well as multilevel polar coded modulation (MLPCM) are investigated. Numerical results provide theoretical performance benchmarks, highlight the potential of MLPCM in achieving close to optimal performance with short codewords, and confirm the discrepancy between the performance in the FBR and that predicted by analysis in the AR. Finally, the meta distribution of the coding rate is derived, providing the percentiles of users that achieve a predefined target rate in a network.

%%The advent of delay-sensitive applications that require ultra high reliability led to Ultra-Reliable Low-Latency Communication (URLLC), as one of the use cases that 5G and beyond 5G must support. These applications are expected to be more abundant {with more stringent} latency requirements {in future networks which will feature dense deployment to improve} availability and coverage. This necessitates studying the performance of large-scale networks in providing low-latency communication. Thus, using the theory of coding in the finite blocklength regime and stochastic geometry, this paper studies the performance of short codes in a large-scale downlink network, which is compatible with low-latency communications. Theoretical limits {for achievable rates and outage are derived, and the achievable performance} by practical modulation scheme and multilevel polar coded modulation (MLPCM) is investigated. Results show that MLPCM achieves rates which are close to the theoretical limits, highlighting the potential of MLPCM in realizing URLLC while achieving close to {optimal} performance.  {Also, results} show {that works} in the literature {overestimate the achievable rate} and underestimate the outage probability, {since they rely on infinitely long codes}. Finally, the meta distribution of the coding rate is derived, {providing} the percentiles of users that achieve a target {performance} over an arbitrary, but fixed, realization of the network.

\end{abstract}


\IEEEpeerreviewmaketitle

\input{Introduction}

\input{SystemModel}

\input{Rate_Analysis}
\input{Outage_Analysis}

%\input{Autoencoder}

%\section{Autoencoder}\label{sec:Autoencoder}

%\section{Results}\label{sec:Results}
\vspace{-0.5cm}
\section{Conclusion}\label{Conclusion}

{This paper investigates} {different aspects of} the performance of a large-scale DL network in the FBR using {the theory of coding in the FBR} in conjunction with stochastic geometric tools. {We start} with rate analysis where {we derive} the average coding rate using Gaussian codebooks and $M$-ary QAM constellation, and {investigate the} practical scheme of MLPCM {in the FBR showing that its performance approaches the theoretical benchmarks}. {We also study} the rate outage probability, reliability, and the coding rate meta distribution, where {we provide} bounds and approximations. From the results, we conclude that the performance analysis provided under an AR provides a misleading {(overestimated)} performance and cannot be used to characterize the performance in the FBR. {T}he {performance in the FBR provides accurate rate and FER characterization,} and approaches the AR performance at sufficiently large blocklengths $>2048$ and high FER $>10^{-2}$. The developed expressions for the average coding rate, the rate outage probability, and the coding rate meta distribution {are fairly tight as shown by} our numerical results. In conclusion, the results are relevant for the theoretical analysis of large-scale networks in the FBR, and the theoretical performance can be approached using practical transmission schemes such as MLPCM. The results in this paper can be extended to characterize the performance of networks under different multiple access schemes as in \cite{Noma4,Mu_MIMO,RS}.
\appendices
\input{Appendix3}
\input{AppendixA_new}

\input{Appendix2}

\bibliographystyle{IEEEtran}
\bibliography{IEEEabrv}

\ifCLASSOPTIONcaptionsoff
  \newpage
\fi

\end{document}

%% file: Introduction.tex
\section{Introduction}
% Motivation

The proliferating applications and diverse use-cases of beyond 5G (B5G) networks enforce stringent key performance indicators that cannot be realized via conventional coding schemes~\cite{Short_Motivations}. {For instance, mission critical applications require sub-$1\ \text{msec}$ latency that cannot be achieved with conventional long codes~\cite{5G_2}. Moreover,  many Internet of Things (IoT) devices have an extreme low-power consumption profile that cannot support the complexity and long transmission intervals of conventional long codes~\cite{IoT}. Consequently, there is a {paradigm} shift towards {using} short codes to comply with the stringent latency and power consumption constraints of the foreseen B5G networks.}\footnote{{In a network, there are four dominant delays: transmission delay, propagation delay, processing delay, and queuing delay. As the network becomes denser, the transmission delay can dominate the propagation delay which motivates using short codes.}} Such reduced code length comes at the expense of inevitable errors and lower information coding rates. The tolerance to the induced errors differs with the application, which can be lower than $10^{-9}$ for ultra-reliable low-latency communications (URLLC)~\cite{5G_1}. {This motivates the study of the finite blocklength coding theory~\cite{polyanski, BlockFadingPolyankiy}, which characterizes the information coding {rate} as a function of the code length and frame error rate (FER).} Therefore, for a context-aware short code design in B5G networks, it is of primary importance to extend finite blocklength coding theory to interference{-prone} scenarios that account for the intrinsic large-scale and dense deployments {of future networks}.

%there is a fundamental need to extend finite blocklength coding theory interference-aware mathematical models to characterize the performance of  in finite blocklength regime (FBR). 

%Nevertheless, the impact of the short block length in large-scale cellular networks is not yet fully characgerized 

%Nowadays, communication networks are required to not only provide better coverage, more connectivity, and higher data rates, but also to enable ultra reliable low latency communication (URLLC)~\cite{5G_2}. As such, they can support delay-sensitive applications like self-driving cars, remote diagnosis, and augmented reality (AR)-assisted surgery, {all of which need to transmit delay-sensitive data to enable accurate and safe operation.} Low latency has been defined for 5G as achieving a latency less than $1\ \text{msec}$ while simultaneously achieving ultra high reliability up to a frame error rate (FER) of $10^{-5}$, whereas 6G is expected to support even higher reliability up to a FER of $10^{-9}$~\cite{5G_1}. In this context, latency is governed by several factors, one of which is the transmission delay which is a function of the code length. Therefore, short codewords should be used, as in Internet-of-Things applications which use codelengths of $512$ up to $4096$ bits~\cite{IoT}{. This use of short codes results in a trade-off between latency, reliability, and data rate.}
Performance of large-scale cellular networks is well investigated {using the} classical Shannon capacity~\cite{paper1_SG,paper2_SG,paper3_SG,SG,sawy,sawy2}, defined as the maximum achievable rate such that the error probability vanishes as the code length increases \cite{Shannon}. Such idealistic scenario is hereafter denoted as the asymptotic regime (AR). The AR analysis in~\cite{paper1_SG,paper2_SG,paper3_SG,SG,sawy,sawy2} is not applicable for networks operating in the finite blocklength regime (FBR). Short codes (e.g., 128 {symbols}) violate the AR assumptions and lead to inevitable errors as a cost for {constraints on} delay and/or power consumption. {Thus,} the code length $n$ of the FBR is a critical design parameter that imposes a delicate trade-off between FER ($\epsilon$) and information coding rate. To characterize such trade-off, the work of Polyanskiy \textit{et al.} in~\cite{polyanski,BlockFadingPolyankiy} finds the maximum achievable rate as function of the code length $n$ and FER $\epsilon$ for a point-to-point link.  %{This regime of operation (short codes) is known as the finite block-length regime (FBR), and we refer to the former regime (infinitely long codes) as the asymptotic regime (AR).} 

%The multilevel coding provided improvement in the performance of modulation schemes which is shown in~\cite{multilevel1, multilevel2}. Coding is then add to enhance the performance of the system by using error correcting codes like polar codes. So, in~\cite{multilevel_polar}, a joint optimization of binary polar coding and multilevel coding approach is proposed for M-ary pulse amplitude modulation (PAM) is proposed. Simulation results are provided for AWGN channel. Additionally, in~\cite{ multilevel_polar2}, a multilevel polar coded modulation (MLPCM) is designed for hybrid automatic repeat request (HARQ) protocol. Also, a set of algorithms to design throughput maximizing MLPCM for the successive cancellation decoding is introduced. Furthermore, they propose a rate matching algorithm to find the best rate for successive cancellation list decoders. The resulting codes provide throughput close to the capacity with low decoding complexity. In this paper, the MLPCM is used to validate our theoretical results.

Motivated by the {need to migrate} towards the FBR, the work in \cite{polyanski,BlockFadingPolyankiy} is extended for several use cases in~\cite{Relaying,Noma1,RS,Mu_MIMO, URLLC_1,URLLC_2,URLLC_4,URLLC_5,URLLC_7,URLLC_8,URLLC_9}. For instance, the work in~\cite{Relaying} extends the FBR analysis {to} relaying channels. The achievable rate for non-orthogonal multiple access in the FBR is studied in \cite{Noma1}. The impact of the FBR on rate splitting is characterized in \cite{RS}. The performance of multi-user MIMO, massive-MIMO, and cell-free MIMO in the FBR are investigated in \cite{Mu_MIMO}, \cite{URLLC_1}, and \cite{URLLC_2}, respectively. The authors in \cite{URLLC_4,URLLC_5} optimize the blocklength and other network parameters (e.g., transmit power or transmitter location) to minimize the decoding error probability subject to a latency constraint. Works was also done in~\cite{URLLC_7,URLLC_8,URLLC_9} to maximize the achievable rate, improve latency, enhance reliability, and/or provide a secure communication in the FBR. However, the models in~\cite{Relaying,Noma1,RS,Mu_MIMO, URLLC_1,URLLC_2,URLLC_4,URLLC_5,URLLC_7,URLLC_8,URLLC_9} are limited to small scale networks. To the best of the authors' knowledge, performance characterization of large-scale networks {in the FBR} is still an open problem.

% Contribution

Using the coding theory in the FBR in conjunction with stochastic geometry tools, this paper intends to {contribute to} the aforementioned research gap, {as an extension of \cite{mypaper}}. In particular, this paper develops novel mathematical {analysis} to characterize the trade-off between the code length $n$, the FER $\epsilon$, and the information coding rate $R$ in large-scale downlink (DL) networks using orthogonal multiple access (OMA). To this end, achievable rates and outage {in the FBR} are characterized for Gaussian codebooks as well as practical modulation schemes. 
Multi-level polar-coded modulation (MLPCM) \cite{multilevel_polar,multilevel_polar2,5G_MLPCM_2,5G_MLPCM_3}, is presented as a practical validation {for} the obtained theoretical achievable rates.\footnote{Polar codes is standardized for 5G in 3GPP Release 16 Specification $\#$ 38.212\cite{3gpp1}, and MLPCM was proposed to be used in the future generations in \cite{future_MLPCM1,future_MLPCM2}} The contributions of this paper are summarized as follows.

The paper analyzes the performance of a large-scale DL network in the FBR {as a function of $n$ and $\epsilon$} in terms of:
\begin{itemize}
\item the average coding rate under a random distance between the user and its serving BS, {under Gaussian codebooks},
\item the average coding rate using QAM constellations, 
\item the rate outage probability and reliability bounds and approximations, {and}
\item the approximation of the coding rate meta distribution.
\end{itemize}
{Additionally, the paper simulates the average coding rate achieved by MLPCM for different QAM modulation orders in comparison with the obtained theoretical benchmarks.} All results are validated via Monte Carlo simulations. It is shown that MLPCM achieves rates close to the {derived} theoretical benchmarks. It is also shown that the derived outage probability bounds are tight and coincide with simulation results. Additionally, it is shown that the {proposed approximation of the coding rate} meta-distribution {is fairly} tight under moderate and high SINR. Throughout the paper, we conduct qualitative and quantitative comparisons with the performance in the AR {to motivate the importance of analyzing the performance in the FBR and demonstrate the discrepancy with AR results}. We also investigate the effect of various network parameters in the FBR.

% Sections Outline
This paper is organized as follows. In Sec.~\ref{sec:system_model}, the system model and assumptions of the analysis are presented. In Sec.~\ref{sec:Rate_Analysis}, the average coding rate of a large-scale DL network in the FBR is derived, {under Gaussian codebooks} and under a constraint of {QAM constellations}. Then, the {theoretical results are numerically evaluated and compared with the performance of} MLPCM as a practical scheme. In Sec.~\ref{sec:Outage_Analysis}, bounds on the rate outage probability are derived, {followed by} the characterization of the reliability, the meta distribution of the coding rate, {numerical evaluations}, and a detailed discussion. Finally, the paper is concluded in Sec.~\ref{Conclusion}.

%% file: SystemModel.tex
%\vspace{-0.7cm}
%\vspace{-0.4cm}
\section{System Model}\label{sec:system_model}

{Consider a single-tier OMA large-scale DL network with universal frequency reuse and no intra-cell interference. The base stations (BSs) are located according to a} Poisson point process (PPP) %\footnote{The network is approximated as a PPP with interference exclusion region of radius $r_0$, i.e., the distance between a UE and its serving BS is $r_0$ and there is no interfering BS within radius $r_0$ from the UE~\cite{sawy}. } 
 $\Psi \subset \mathbb{R}^2$ with intensity $\lambda\ \text{BS}/\text{km}^2 $. %The users equipment  (UEs) follow an independent PPP with density $\lambda_{\rm UE}$ $\text{UEs}/\text{km}^2$. 
{The network applies} universal frequency reuse with one user equipment (UE) served per BS at a given {resource block}. Each UE is served by its geographically closest BS, where the intended distance between a UE and its serving BS is denoted as $r_0$. The distances to the interfering BSs ordered with respect to the intended UE are denoted as $r_1,\ r_2,\dots,\ r_i,\dots$ such that $r_{i+1}>r_{i}$. {For the sake of simple presentation, the set $\Tilde{\Psi}\in\mathbb{R}$ is defined {as} the BSs distances to the desired UE.} Hence, the received signal at the desired UE is given by
%\vspace{-0.15cm}
\begin{equation}\label{eq:system_model}
   y = \sqrt{\mathcal{P}} h_0  r_0^{-\eta/2} s_0 + \underset{ I_{agg}}{\underbrace{\sum_{r_i\in\Tilde{\Psi}\setminus \{r_0\}}\sqrt{\mathcal{P}} h_i  r_i^{-\eta/2} s_i} } +w,
\end{equation}
where $\mathcal{P}$ is the transmit power of the {BSs}, $h_0$ (resp. $h_i$) {represents the channel fading}  of the intended (resp. $i^{th}$ interfering) channel, $\eta$ is the path loss exponent, $s_0 $ (resp. $s_i$) is a unit average power codeword symbol transmitted by the serving (resp. $i^{th}$ interfering) BS, $I_{agg} $ is the aggregate interference from all other BSs, and $w\sim \mathcal{CN}(0,\sigma_w^2)$ is circularly symmetric complex Gaussian noise with zero mean and variance $\sigma_w^2$. {We model $h_0$ and $h_i$ to be circularly symmetric complex Gaussian with zero mean and unit variance (modeling Rayleigh fading). An independent and identically distributed block fading channel model is assumed where the channel coefficients $h_0$ and $h_i$ remain constant for a duration of {$L$} consecutive symbols and change to independent realizations at the end of each symbol interval.}

 %The interference term $I_{agg} $ is the sum of interference signals received from all non-serving BSs, and is given by 

%\vspace{-0.2cm}
%\begin{equation}\label{eq:I_gg}
%    I_{agg} =\sum_{r_i\in\Psi\setminus \{r_0\}}\sqrt{\mathcal{P}} h_i  r_i^{-\eta/2} s_i ,
%\end{equation}
%where we assume equal power $\mathcal{P}$ across BSs, $h_i\sim\mathcal{CN}(0,1)$ is the channel gain between non-serving $\mathrm{BS}_i$ and the {UE}, and $s_i$ is the unit-power codeword symbol transmitted by the $\mathrm{BS}_i$. 

%Note that the association distance $r_0$ follows a Rayleigh distribution with the following probability density function
%\vspace{-0.2cm}
%\begin{align}\label{eq:r_o}
 %   f_{r_0}(r_0)=   2 \pi \lambda r_0 e^{-\pi \lambda r_0^2},\ \  0<r_0<\infty. 
%\end{align}
%\vspace{-0.2cm}

%\begin{figure}
 %   \centering
  %  \includegraphics[width=0.45\linewidth]{graph3.eps}
   % \caption{A large-scale network with $\lambda\hspace{-0.07cm}=\hspace{-0.07cm}1\ \text{BS}/\text{Km}^2$ and $\lambda_{\text{\tiny{UE}}}=10\ \text{UEs}/\text{Km}^2$ in an area $=100\ \text{Km}^2$.}
%    \label{fig:system}
%\end{figure}
%A large-scale network with BS intensity $\lambda=1\ \text{BS}/\text{Km}^2$ and UE intensity $\lambda_{\text{\tiny{UE}}}=10\ \text{UEs}/\text{Km}^2$ on an area of $100\ \text{Km}^2$.

%{The channel is assumed to be time-varying with a coherence duration of $L$ symbols. }

The serving BS wants to send information to the UE using codewords from a code with the rate $R$ and length $n$ symbols, where $n$ can be defined following a latency and/or energy consumption constraint. To ensure that the channel remains constant during the transmission of $n$ symbols, we require $n= L/l$ for some integer $l$ {\cite{BlockFading}.} {This ensures that the channels remain constant during the transmission, but makes consecutive transmission blocks have identical channels. Nonetheless, by considering the sequence of blocks consisting of the first block (of $n$ symbols) in each coherence interval (of $L$ symbols), we have independent channels and we can derive the average performance over all transmissions. This is because the remaining sequences of blocks (such as the second frame in each coherence interval) have identical statistics.}

{To this end, two modes of operation are considered.
The first mode assumes knowledge of channel-state information (CSI) at the transmitter and the receiver in the form of SINR availability.\footnote{Assuming SINR knowledge benchmarks the performance of the network.} {Under this} mode, we investigate the average coding rate $\mathbb{E}\{R\}$ as a function of the FER and codelength (see Sec. \ref{sec:Rate_Analysis}). The second mode assumes that the CSI is unknown at the transmitter but known at the receiver. Under this scenario, we {derive the} outage probability and utilize the coding-rate meta distribution to characterize the percentile of users achieving a given transmission reliability (See Sec. \ref{sec:Outage_Analysis}).}

%% file: Rate_Analysis.tex
\vspace{-0.25cm}
\section{Average Rate Analysis}\label{sec:Rate_Analysis}

{Considering the first mode of operation, we assume that the SINR is available at the transmitter and the receiver. Hence, the coding rate can be adapted based on the SINR in each transmission block so that the transmission using codelength $n$ is successful with a desired FER $\epsilon$. This rate adaptation can be realized using the following lemma, which characterizes the maximum achievable rate over an additive white Gaussian noise (AWGN) channel in the
FBR \cite{polyanski}.}
\begin{lemma}{\textsc{(\cite{polyanski})}}\label{lem:poly}
For an AWGN channel with {signal to noise power ratio (SNR)} $\alpha$, blocklength $n$, and FER {$\epsilon\in(0,0.5)$}, the maximum coding rate is approximated as{\footnote{{This approximation is shown to be tight at blocklengths ($>100$) and FER $>10^{-6}$} in \cite{polyanski}.} }
\begin{equation}\label{eq:Cap_polyankiy}
    R_{n,\epsilon}(\alpha)= C_{\infty,0}(\alpha)-\frac{\sqrt{V(\alpha)} Q^{-1}(\epsilon)}{\sqrt{n}}{+\frac{1}{2n} \log_2 n},
\end{equation}%% check
where  $C_{\infty,0}(\alpha) = \log_2(1+\alpha)$ is the AWGN {channel} capacity in the AR, $V(\alpha)= \frac{\alpha(\alpha+2)}{(\alpha
+1)^2} \log_2^2(e)$ is {known as} the channel dispersion, and $Q^{-1}(\cdot)$ is the inverse of the Q-function. %\ac{The $O(1)$ term is negligible when $n>100$ at
%$\epsilon=\left\{10^{−3}, 10^{−6}\right\}$ \cite{polyanski}, and \eqref{eq:Cap_polyankiy} is used henceforth.}
\end{lemma}

{In this mode, due to the random changes of the channels $h_0$ and $h_i$, we characterize performance by the average coding rate subject to a target FER $\epsilon$ and codelength $n$.} However, Lemma~\ref{lem:poly} is derived for {an AWGN channel}, and hence, cannot be directly applied to a large-scale network due to the additional non-Gaussian interference term $I_{agg}$ in \eqref{eq:system_model}. To overcome this, we utilize the equivalence in distribution (EiD) approach to express the interference as a conditionally Gaussian random variable~{\cite{EID,eid2,sawy}}, as described next.

\subsection{Conditional Gaussian Representation}
Using the EiD approach, we represent the aggregate interference for a given $r_0$ as follows
\begin{align}\label{eq:I_eq}
 I_{agg} {\overset{\text{eid}}{=}} \sqrt{\mathcal{B}} {g},
 \end{align}
where {$\overset{\text{eid}}{=}$} denotes the EiD, ${g}\sim \mathcal{CN}(0,1)$ and $\mathcal{B}>0$ is a positive random variable independent of $g$ with a {probability density function} (PDF) whose Laplace transform (LT) {is given by}~\cite{sawy},
\begin{align}
\label{LapTrans}
   { \mathcal{L}_\mathcal{B}(u)=\exp\left\{ \sum_{k=1}^{\infty} (-1)^k 2 \pi \lambda r_0^2 \left(\frac{\mathcal{P}}{r_0^\eta}\right)^k \frac{\mathbb{E}\big\{|s|^{2k}\big\}u^k}{(\eta k -2)k!} \right\}.}
\end{align}
{where $s$ is the transmitted symbol from an arbitrary codebook by an arbitrary BS and follows the same distribution as $s_0$ and $s_i$.}

The EiD in \eqref{eq:I_eq} is proved by showing that the product $ \sqrt{\mathcal{B}} {g}$ has the same characteristic function as $I_{agg}$ \cite{sawy}. Using \eqref{eq:I_eq}, the equivalent system model is expressed as 
\begin{align}
y=\sqrt{\mathcal{P}} h_0  r_0^{-\eta/2} s_0 +  \sqrt{\mathcal{B}} {g}  +w,\nonumber 
\end{align}
where for a given $\mathcal{B}$, the lumped term $\sqrt{\mathcal{B}} {g}+w$ is Gaussian with zero mean and variance $\mathcal{B}+\sigma_w^2$. Hence, the resulting conditional SINR, given $\mathcal{B}$ and $h_0$, is given by 
%\vspace{-0.15cm}
\begin{equation}\label{eq:SINR}
 {\Upsilon}=\frac{|h_0 |^2  \mathcal{P} r_0^{-\eta}}{\mathcal{B}+\sigma_w^2}.   
\end{equation}
%\vspace{-0.15cm}
The random variable $\mathcal{B}$ {in \eqref{eq:SINR}} augments the noise power with the aggregate network interference power, for any constellation {of $s_0$ and $s_i$. If $s_0$ and $s_i$ is Gaussian distributed} (Gaussian codebooks), it can be shown that \eqref{LapTrans} simplifies to \cite{sawy}
\begin{align}
\label{LapTrans_Gaussian}
    \mathcal{L}_{\mathcal{B}_g}(u)=\exp\left\{\frac{- 2 \pi \lambda u \mathcal{P} r_0^{2-\eta}}{\eta-2}  {}_2 F_1\left(1,1-\frac{2}{\eta};2-\frac{2}{\eta};-\frac{u \mathcal{P}}{r_0^\eta}  \right) \right\}.
\end{align}

%%Then, the interference-aware average rate can be obtained via additional deconditioning steps on the random variables $h_0$, $\mathcal{B}$, and $r_0$. 
Since at this point, interference is modeled as conditionally Gaussian, and since Lemma \ref{lem:poly} provides a tight approximation when noise is Gaussian, then the tightness of the approximation provided in Lemma \ref{lem:poly} holds here too. Thus, from this point onward, we focus on analyzing performance using $R_{n,\epsilon}$ as a surrogate for the maximum achievable rate in the FBR given that this is a tight approximation. Under this framework, the maximum average coding rate is defined next.
\vspace{-0.2cm}
{\begin{definition}\label{def_avg_Rate}
Given a DL large-scale network with a distance $r_0$ between the desired UE and the serving BS as defined in (1), assuming the SINR knowledge is available at the transmitter and the receiver, we define the maximum average coding rate in the FBR with blocklength $n$ and FER $\epsilon$ with $\alpha_0=\frac{\mathcal{P}r_0^{-\eta}}{\sigma_w^2}$, as $R_{n,\epsilon} (\alpha_0 )=\mathbb{E}_{h_0,\mathcal{B}} \{R_{n,\epsilon} (\Upsilon)\}$, where $\Upsilon$ is the SINR defined in \eqref{eq:SINR} and $R_{n,\epsilon} (\Upsilon)$ is as defined in Lemma~\ref{lem:poly}.
\end{definition}}

{Next, we use Lemma \ref{lem:poly} and Def. \ref{def_avg_Rate} to characterize the average coding rate for a large-scale network in the FBR, first using Gaussian codebooks, and then under finite constellations (QAM). Then,} we present a practical MLPCM to validate the achievability of the derived theoretical benchmarks.

%In this section, the average coding rate for a large-scale network in the FBR is first derived then the average coding rate under finite constellation is investigated. A practical (MLPCM) is then presented to validate the achievability of the aforementioned interference-aware theoretical benchmarks.
\input{Ergodic_Capacity}
\input{Modulation}
\input{PolarCodes}
\input{Rate_Numerical_Results_Modified}

%% file: Ergodic_Capacity.tex
\subsection{Average Coding Rate in the Finite Block-Length Regime}\label{sec:outage_capacity}

%\vspace{-0.4cm}

By virtue of the EiD, the AWGN results in Lemma~\ref{lem:poly} can be {now} extended to large-scale networks. For a given intended UE distance $r_0$, the FBR average coding rate is characterized in the following theorem.
%\vspace{-0.3cm}
\begin{theorem}\label{Theorem:Avg_coding_Network}
The maximum average coding rate of the large-scale network modeled by~\eqref{eq:system_model} {with Gaussian codebooks,} blocklength $n$, FER $\epsilon$, and distance $r_0$ between the {UE and its serving BS} is given by
\begin{align}\label{eq:Capacity_FB_r_0}
    \mathcal{R}_{n,\epsilon}({\alpha_0})=\mathcal{C}_{\infty,0}({\alpha_0})-\frac{\mathcal{V}({\alpha_0}) Q^{-1}(\epsilon)}{\sqrt{n}} {+\frac{1}{2n} \log_2 n},
\end{align}
{where $\mathcal{C}_{\infty,0}({\alpha_0})=\int_{0}^{\infty} \exp\left(-\frac{2^{ c}-1}{{\alpha_0}}\right) \mathcal{L}_{\mathcal{B}_g}\left\{\frac{2^{ c}-1}{\mathcal{P} r_0^{-\eta}}\right\} d c,$ $
\mathcal{V}({\alpha_0})=\int_{0}^{\log_2(e)} e^{\frac{-z(v)}{{\alpha_0} } }\mathcal{L}_{\mathcal{B}_g}\hspace{-0.05cm}\left\{\frac{r_0^{\eta}}{\mathcal{P} }z(v)\right\} d{v}$, $\alpha_0=\frac{\mathcal{P} r^{-\eta}_0}{\sigma_w^2}$, $z(v)=\sqrt{\frac{1}{1-\frac{{v^2}}{\log_2^2(e)}}}-1$,
and $\mathcal{L}_{\mathcal{B}_g}\{\cdot\}$ is given in \eqref{LapTrans_Gaussian}.}
%\begin{align}
%\mathcal{C}_{\infty,0}(\alpha,r_0)&=\int_{0}^{\infty} \exp\left(-\frac{2^{ c}-1}{\alpha r_0^{-\eta}}\right) \mathcal{L}_\mathcal{B}\left\{\frac{2^{ c}-1}{\mathcal{P} r_0^{-\eta}}\right\} d c,\nonumber\\ 
 %\mathcal{V}(\alpha,r_0)&=\int_{0}^{\log_2(e)} e^{\frac{-r_0^{\eta}}{\alpha } z(v)}\mathcal{L}_\mathcal{B}\hspace{-0.05cm}\left\{\frac{r_0^{\eta}}{\mathcal{P} }z(v)\right\} d{v}, \nonumber
%\end{align}
\end{theorem}
 \vspace{-0.2cm}
\begin{proof}
{By virtue of the EiD approach, Lemma~\ref{lem:poly} is applicable to large-scale networks by replacing $\alpha$ by $\Upsilon$ where $\Upsilon$ is defined in \eqref{eq:Cap_polyankiy}, which leads to a maximum coding rate $R_{n,\epsilon}(\Upsilon)$ under a given $h_0$ and $\mathcal{B}$. The average rate is then $ \mathcal{R}_{n,\epsilon}({\alpha_0}) = \mathbb{E}\{R_{n,\epsilon}(\Upsilon)\}$, where the averaging is over $h_0$ and $\mathcal{B}$.} {Thus, we need to average $C_{\infty,0}(\Upsilon)$ and $\sqrt{V(\Upsilon)}$ over $\mathcal{B}$ and $h_0$.} The term $\mathcal{C}_{\infty,0}({\alpha_0})=\mathbb{E}
\{C_{\infty,0}(\Upsilon)\}$ {is the average capacity in a large-scale network in the AR with CSI available at the BSs, which was} derived in~\cite{r_o}. {It remains to find } $\mathbb{E}
\{\sqrt{{V}(\Upsilon)}\}$ which is shown to be equal to $\mathcal{V}(\alpha_0)$ in App. \ref{sec:Appendix3}.
\end{proof}
Theorem 1 governs the maximum average rate at a fixed distance $r_0$ between the UE and its serving BS, considering the stochastic locations of the interfering BSs and the Rayleigh {fading environment}. The following corollary generalizes Theorem 1 by accounting for the randomness of $r_0$.
\begin{Corollary}\label{Collary:Avg_Coding_Rate}
The average coding rate of the large-scale network modeled by~\eqref{eq:system_model} with blocklength $n$, FER $\epsilon$, {and Gaussian codebooks} is given by
\begin{align}
    \mathcal{\Bar{R}}_{n,\epsilon} = \mathcal{\bar{C}}_{\infty,0}-\frac{\mathcal{\bar{V}} Q^{-1}(\epsilon)}{\sqrt{n}}{+\frac{1}{2n} \log_2 n},\nonumber
\end{align}
where $\mathcal{ \Bar{C}}_{\infty,0}= \int_{0}^{\infty} \mathcal{ {C}}_{\infty,0}({\alpha_0}) f_{r_0}(r_0) d r_0 $
and $ \mathcal{ \Bar{V}}= \int_{0}^{\infty} \mathcal{ {V}}({\alpha_0}) f_{r_0}(r_0) d r_0$.
\end{Corollary}
\begin{proof}
The average capacity and average channel dispersion are derived by directly averaging~\eqref{eq:Capacity_FB_r_0} with respect to $r_0$ which has the following probability density function{\cite{r_o}} 
\begin{align}\label{eq:r_o}
    f_{r_0}(r_0)=   2 \pi \lambda r_0 e^{-\pi \lambda r_0^2},\ \  0<r_0<\infty. 
\end{align}
\end{proof}
{As} can be seen in Theorem~\ref{Theorem:Avg_coding_Network} and Corollary~\ref{Collary:Avg_Coding_Rate}{, the average coding rate of the large-scale network in the FBR is} the average capacity in the {AR, plus} a penalty term (loss) which is a function of the blocklength $n$, and FER {$\epsilon$, in addition to the {average SNR $\alpha_0$} and the stochastic geometry of the network, manifested in the Laplace transform term}. This penalty term implies a trade-off between the reliability and the average rate: the higher the reliability requirement (low FER), the lower the rate. However, for a given average achievable rate, the reliability can be increased ($\epsilon$ decreased) by increasing the blocklength, converging to the AR performance as $n$ grows to infinity. 

%\vspace{-0.1cm}
{Theorem \ref{Theorem:Avg_coding_Network} expresses the average coding rate under Gaussian codebooks, which is of theoretical relevance}. Characterizing the performance in the FBR {when using} a finite constellation {is of practical interest} and is discussed {next}.

%% file: Modulation.tex
%\vspace{-0.2cm}
\subsection{Average Rate Using Finite Constellations}\label{sec:Constellation}
%\vspace{-0.02cm}
This section focuses on achievable rates in the FRB {under transmission using} standard finite constellations. {The maximum achievable rate using an $M$-ary constellation over an AWGN channel in the FBR was given
in~\cite{modulation}. Following the same methodology as in Sec.~\ref{sec:outage_capacity}, we use the EiD approach to extend the results of~\cite{modulation} to large-scale networks.} 
%We denote the symbols of the $M$-ary constellation as $\{s_1,s_2,\ldots,s_M\}$}.

Consider an encoder that maps messages from a message set into a length $n$ sequence of symbols chosen from an $M$-ary constellation consisting of $M$ complex-valued symbols $s_1,s_2,...,s_M$. {Using an $M$-ary QAM constellation, denote by $\mathcal{C}^{\text{\tiny(M)}}_{\infty,0}(\alpha_0)$ the average capacity achieved in the AR, by ${\mathcal{ V}}^{\text{\tiny{(M)}}}(\alpha_0)$ the average square root channel dispersion, and by $\mathcal{R}^{\text{\tiny(M)}}_{n,\epsilon}(\alpha_0)$ the average coding rate under a blocklength $n$ and FER $\epsilon$.} {To characterize the average capacity of the $M$-ary QAM under FBR, the conditional mutual information $I(s_0;y|h_0,\mathcal{B},r_0)$ has to be averaged with respect to the aggregate interference power $\mathcal{B}$ and the channel gain $h_0$. However, the distribution of the aggregate interference power $\mathcal{B}$ is unknown which leads to an intractable expression. Although such an expression can be evaluated using Monte Carlo simulations, an expression that is amenable to numerical integration is preferable. A similar argument applies for ${\mathcal{ V}}^{\text{\tiny{(M)}}}(\alpha_0)$. Hence, to obtain tractable integral expressions, we adopt an approximation for the distribution of the interference power $\mathcal{B}$ proposed in \cite{Gamma_app_1,Gamma_app_2} which was shown to be a tight approximation. In particular, in \cite{Gamma_app_1,Gamma_app_2}, the distribution of the interference power $\mathcal{B}$ was approximated by a Gamma distribution 
\begin{align}\label{eq:Gamma_app}
    f(x;q,\theta)=\frac{x^{q-1} e^{\frac{-x}{\theta}}}{\Gamma(q) \theta^q},\ \ \ \ \ \ \ \ \   x>0
\end{align}
where $q=\frac{4 \pi \lambda r_0^2 (\eta-1)}{(\eta-2)^2}$ is the shape parameter, and $\theta=\frac{(\eta-2)\mathcal{P}}{2(\eta-1)r_0^\eta}$ is the scale parameter. The approximation is based on a moment-matching approach. The first two moments of $\mathcal{B}$ can be obtained from \eqref{LapTrans} as 
\begin{align}
\mathbb{E}\{\mathcal{B}\}&=\left.\frac{d \mathcal{L}_\mathcal{B}(u)}{du}\right|_{u=0} = \frac{2 \pi\lambda r_0^{2-\eta} \mathcal{P}}{ (\eta-2)}\nonumber\\
\mathbb{E}\{\mathcal{B}^2\}&= \left.\frac{d^2 \mathcal{L}_\mathcal{B}(u)}{du^2}\right|_{u=0}=\left(\frac{2 \pi\lambda r_0^{2-\eta} \mathcal{P}}{ (\eta-2)}\right)^2+\frac{ \pi \lambda r_0^{2-2\eta} \mathcal{P}^2}{ (\eta-1)}.\nonumber
\end{align}
The scale and shape parameters of the gamma distribution are then obtained as $q=\frac{\mathbb{E}\{\mathcal{B}\}^2}{\mathbb{V}ar\{\mathcal{B}\}}$ and $\theta=\frac{\mathbb{V}ar\{\mathcal{B}\}}{\mathbb{E}\{\mathcal{B}\}}$, where  $\mathbb{V}ar\{\mathcal{B}\} = \mathbb{E}\{\mathcal{B}^2\}-\mathbb{E}\{\mathcal{B}\}^2$.
}

{Fig.~\ref{fig:gamma_vs_exact} numerically demonstrates the accuracy of this approximation by plotting the approximate {CDF} of $\mathcal{B}$ (using \eqref{eq:Gamma_app}) and its exact {CDF} obtained from a numerical inversion {of}~\eqref{LapTrans}. The figure shows that the Gamma approximation in \eqref{eq:Gamma_app} provides a fairly tight approximation for the CDF of $\mathcal{B}$. Note that the Gamma approximation is crucial to maintain the tractability of the analysis. Using this approximation, the average coding rate of a large-scale network in the FBR, at a {given} distance $r_0$ and using $M$-ary QAM constellation, is characterized.}

\begin{figure}
    \centering
    \includegraphics[width=0.55\linewidth]{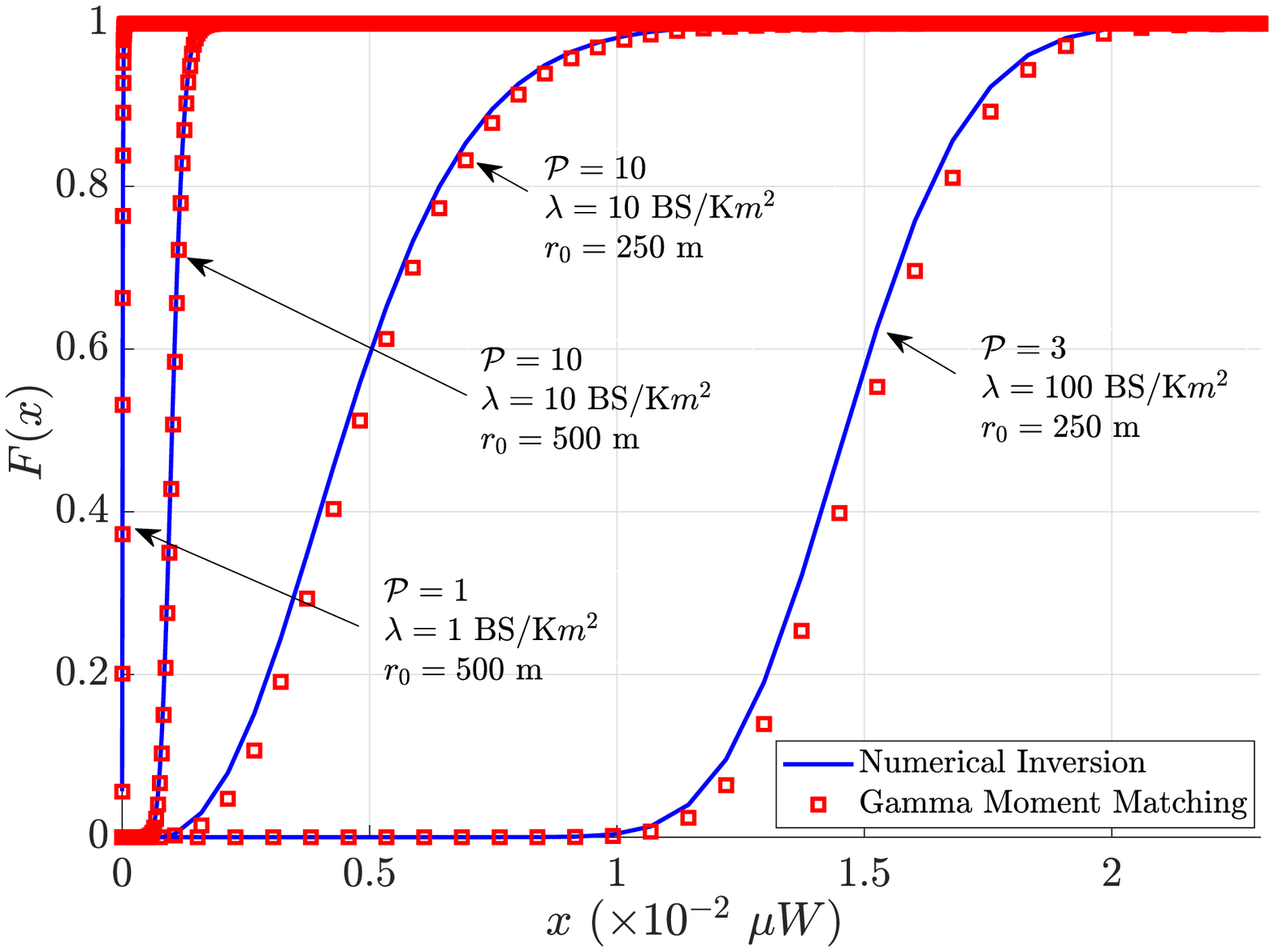}
    \caption{CDF of the interference power {$\mathcal{B}$} and its approximation at different values of\hspace{0.08cm}$\mathcal{P}$,\hspace{0.05cm}$\lambda$ and $r_0$}
    \label{fig:gamma_vs_exact}
\end{figure}
{The maximum achievable rate for a finite constellation is provided in \cite{modulation}. However, the channel dispersion in \cite{modulation} is defined as the variance of the information density ($V=Var(i(x;y))$). This definition is considered to be imprecise according to \cite{polyanski}, where the channel dispersion is defined as the conditional variance of the information density for finite constellation and is defined as the unconditional variance of the information density in the case of using Gaussian signals. Therefore, in this paper, we derive the maximum achievable rate using this definition for large-scale networks in the FBR in the following theorem.}
{
%\vspace{-0.3cm}
\begin{theorem}\label{Theorem:ModulationCapacity}
For a DL large-scale network with {BS} density $\lambda\ \mathrm{BS}/\mathrm{km}^2$, blocklength $n$, FER $\epsilon$, distance $r_0$ between the UE and the serving BS, and an $M$-ary constellation with symbols $\{\underline{s}_m\}_{m=1}^M$ where $\underline{s}_m=[\mathcal{R}e\{s_m\},\mathcal{I}m\{s_m\}]^T$ and $\mathcal{R}e\{s_m\}$ and $\mathcal{I}m\{s_m\}$ are the real and imaginary parts of $s_m$, respectively, if the interference power $\mathcal{B}$ follows a Gamma distribution \eqref{eq:Gamma_app}, then the {average} achievable coding rate can be expressed by
%\vspace{-0.8cm}
\begin{align}\label{eq:modulation}
    \mathcal{R}^{\text{\tiny(M)}}_{n,\epsilon}(\alpha_0)=\mathcal{C}^{\text{\tiny(M)}}_{\infty,0}(\alpha_0)-\frac{\mathcal{{V}}^{\text{\tiny{(M)}}}(\alpha_0) Q^{-1}(\epsilon)}{\sqrt{n}}{+\frac{1}{2n} \log_2 n},
\end{align}
where $\alpha_0=\frac{
\mathcal{P} r_0^{-\eta} }{\sigma_w^2}$,
%\vspace{-0.5cm}
{\begin{align}
  \mathcal{C}^{\text{\tiny(M)}}_{\infty,0}(\alpha_0)&=\log_2(M)- \frac{1}{M \pi}\sum_{m=1}^{M}\int_{0}^{\infty}\int_{{\mathbb{R}^2}} e^{-\|\underline{t}\|^2} {g_m(\underline{t})}  f_{{\Upsilon}}(\Upsilon|r_0) d\underline{t} \ d \Upsilon \label{eq:Capacity_mod_FB}\\
{\mathcal{ V}}^{\text{\tiny{(M)}}}(\alpha_0)&=\int_{0}^{\infty}\sqrt{V_{\text{\tiny M} }(\Upsilon)} f_{{\Upsilon}}(\Upsilon|r_0)\ d\Upsilon \label{eq:V_averaged_mod}\\
    {V}_{\text{\tiny{M}}}(\Upsilon)&=\sum_{m=1}^{M}\left(\int_{{\mathbb{R}^2}} \frac{e^{-\|\underline{t}\|^2}}{M \pi}  {g_m(\underline{t})^2} d\underline{t} 
  -\left( \int_{\mathbb{R}^2} \frac{e^{-\|\underline{t}\|^2}}{M \pi}  {g_m(\underline{t})}d \underline{t} \right)^2 \right)\label{eq:V_mod}
\end{align}}
{where $g_m(\underline{t})=\log_2\left( \sum_{l=1}^{M} e^{ -2\sqrt{{\Upsilon}}\underline{t}^T(\underline{s}_m-\underline{s}_l)-  {\Upsilon}\|\underline{s}_m-\underline{s}_l\|^2} \right)$, $\underline{t}=[t_1,t_2]^T$ is a 2-D real-valued vector, $\Upsilon$ is the SINR given in~\eqref{eq:SINR}}, and  
\begin{align}\label{eq:pdf_SINR}
f_{{\Upsilon}}(\Upsilon|r_0)=\frac{\exp\left\{\frac{-\Upsilon \sigma_w^2}{\mathcal{P}r_0^{-\eta}}\right\}}{\mathcal{P}r_0^{-\eta}} \frac{(1+\theta r_0^{\eta} \Upsilon/\mathcal{P})\sigma_w^2+q\theta}{(1+\theta r_0^{\eta} \Upsilon/\mathcal{P})^{q+1}},  0\leq \Upsilon\leq \infty.
\end{align}
\end{theorem}
}
\begin{proof}
{The average capacity of an $M$-ary constellation {($\mathcal{C}^{\text{\tiny(M)}}_{\infty,0}(\alpha_0)$)} at a fixed $r_0=r$ is derived by averaging the {conditional} mutual information between the transmitted symbol $s_0\in\{s_m\}_{m=1}^{M}$ and the received signal $y$, i.e., $I(s_0;y|\Upsilon,r_0=r)$, provided in \cite{modulation}, with respect to the SINR $\Upsilon$ which is a function of {interference power} $\mathcal{B}$ and fading channel $h_0$ statistics. Similarly, the average square-root channel dispersion (${\mathcal{ V}}^{\text{\tiny{(M)}}}(\alpha_0)$) is derived by averaging the channel dispersion in \eqref{eq:V_mod} with respect to $\Upsilon$. The derivation of \eqref{eq:V_mod} is provided in App.~\ref{appendix}. However, instead of averaging over the distribution of $\Upsilon$ which is unknown, we average over $f_{{\Upsilon}}(\Upsilon|r_0)$ in \eqref{eq:pdf_SINR} which is obtained using the Gamma approximation provided in \eqref{eq:Gamma_app} (See App.~\ref{appendix2})}. 
%interference $\hat{\Upsilon}$ using the Gamma approximation and fading channel $h_0$ statistics
\end{proof}

%The capacity expression for a finite constellation set is the number of bits per symbol with a back-off term that depends on the channel state in terms of SINR and the distance between the constellation point. The expression provided in~\eqref{eq:modulation} as a function of the distance between the user and the serving BS where near users experience a higher performance than far users because of the path loss and the fact that edge users suffer from high interference than center users depending also on the density of the network. 

%The capacity of $M$-ary constellation set is defined as the maximum rate achieved by transmitting symbols of the constellation points through the channel for an infinite blocklength and a vanishing frame error probability and given in~\eqref{eq:cap_Mary}.  

The following corollary generalizes Theorem 2 by accounting for the randomness of $r_0$.
% \vspace{-0.25cm}
 \begin{Corollary}\label{corollary:Modulation}
The average achievable coding rate of the large-scale network modeled by \eqref{eq:system_model} with blocklength $n$, FER $\epsilon$, and using an $M$-ary constellation with symbols $\{s_m\}_{m=1}^M$ is given by
\begin{align}\label{eq:modulation_r0}
    \bar{\mathcal{R}}^{\text{\tiny(M)}}_{n,\epsilon}\approx\bar{\mathcal{C}}^{\text{\tiny(M)}}_{\infty,0}-\frac{\mathcal{\bar{V}}^{\text{\tiny{(M)}}} Q^{-1}(\epsilon)}{\sqrt{n}}{+\frac{1}{2n} \log_2 n},
\end{align}
where $ { \bar{\mathcal{C}}^{\text{\tiny(M)}}_{\infty,0}=\int_{0}^{\infty} {\mathcal{ C}}^{\text{\tiny{(M)}}}_{\infty,0}(\alpha_0) f_{r_0}(r_0) d r_0}$, ${\bar{\mathcal{ V}}^{\text{\tiny{(M)}}}= \int_{0}^{\infty} {\mathcal{ V}}^{\text{\tiny{(M)}}}(\alpha_0) f_{r_0}(r_0) d r_0}$,
and $f_{r_0}(r_0)$ as defined in~\eqref{eq:r_o}. 
\end{Corollary}
%\vspace{-0.45cm}
 \begin{proof}
 This is derived by averaging~\eqref{eq:modulation} with respect to $r_0$. %Hence, the capacity is characterized for a typical user.
 \end{proof}

%The channel dispersion expression proposed is a function of the channel state in terms of SINR and the distance between constellation points. It measures the variation of the information density in a large-scale network environment. 

The {average achievable rate} depends on the distances between symbols in the constellation set, i.e. {$\|\underline{d}_{ml}\|=\|\underline{s}_m-\underline{s}_l\|$} for $m\neq l$, {where} neighboring symbols contribute more to error than far ones. Such distances are in turn affected by the constellation choice. Moreover, {the average achievable rate} is affected by the network parameters $\lambda,\ r_0$ which will be investigated later in the numerical {evaluations} section.

{The achievable rates provided so far are theoretical.} In the next section, we use MLPCM, which was proposed in \cite{future_MLPCM1,future_MLPCM2} to be used in future generations, to validate the achievability of the obtained theoretical benchmarks.

%% file: PolarCodes.tex
\subsection{Multilevel Polar-Coded Modulation}\label{sec:PolarCodes}
\input{MLPCM_figure}
%Polar codes are capacity-achieving codes for a binary-input channel, proposed by Erdal Arıkan in 2009 in~\cite{polarArikan}. The polar encoder is a nonsystematic encoder that uses the SNR of an AWGN channel to calculate the Bhattacharyya parameter~\cite{polarArikan} that is then used in the construction of a polar code with rate $k/n$ where $k$ is the number of information bits and $n$ is the code length. The code construction is based on multiple recursive concatenations of a short kernel code which transforms the physical channel into virtual outer channels. When the number of recursions becomes large, the $n$ virtual channels polarize into $\approx nC$ high-reliability channels (noiseless) or and $\approx n(1-C)$ low-reliability channels (extremely noisy), where $C$ is the capacity of the underlying AWGN channel. Then, information bits are sent over the most reliable channels, and the inputs of the low-reliability channels are frozen to zero.  Bits sent over the noiseless channels are received correctly, therefore achieving a rate close to capacity. 

% Polar encoding
%\subsection{MLPCM Transceiver}

{To validate} the theoretical results in Sec.~\ref{sec:outage_capacity} and ~\ref{sec:Constellation}, {we use the MLPCM scheme shown} in Fig.~\ref{fig:MCM}. {The information bits (message) $\mathbf{X}\in\mathbb{F}_2^{k \log_2⁡(M)}$ are fed into a demultiplexer that splits them into $\log_2(M)$ sub-messages $\{x_i \}_{i=1}^{\log_2⁡(M)}$  where $x_i\in \mathbb{F}_2^k$. Each sub-message $x_i$ is encoded using a polar encoder with rate $R_c=k/n$ to obtain codewords $\mathbf{U}\in \mathbb{F}_2^{\log_2(M)\times n}$, where $u_i\in \mathbb{F}_2^n$ is the $i^{th}$ row of $\mathbf{U}$ ($i^{th}$ codeword). Then, the transmitter collects one bit from each of the codewords $u_i$ to form a $\log_2(M)$-bit symbol which is then mapped to a symbol $s_0$ from an $M$-ary constellation such as the $M$-QAM constellation. The symbols are then scaled according to the power constraint in Sec. II. The result of repeating this process for all $n$ codeword symbols of all $\log_2⁡(M)$ sub-messages is one MLPCM codeword of length $n$, which is then transmitted through the channel. At the receiver side, multistage decoding takes place to efficiently recover the message. It breaks the decoding process into $\log_2(M)$ stages.}

{In each stage $i=1,\dots,\log_2(⁡M)$, the receiver feeds the demapper of the $i^{th}$ stage with the received signal $y$ and the submessages recovered from previous stages $(\Tilde{x}_1,\dots,\Tilde{x}_{i-1})$, (which is defined as $\emptyset$ for $i=1$), and the demapper demaps the symbols into bits to obtain $\Tilde{u}_i$, using the multilevel decoding criteria \cite{multilevel_polar}. Then, the demapped bits $\Tilde{u}_i$ are decoded using a polar decoder to obtain $\Tilde{x}_i$  using the successive cancellation algorithm provided in~\cite{SCD}. Finally, all sub-messages are fed to a multiplexer to form the message $\mathbf{\Tilde{X}}$. This MLPCM will be incorporated in the large-scale DL network Monte Carlo simulator to validate the achievability of the theoretical rates obtained in Sec.~\ref{sec:outage_capacity} and~\ref{sec:Constellation}.}
%

%For the multilevel mapping and demapping, the constellation set is partitioned into subsets of lower-order constellation sets. At each level, the constellation points are partitioned into halves until the subsets have a single constellation point and each subset at each level is mapped to a zero or a one. As an example, in Fig.~\ref{fig:8QAM_ML}, the 8-QAM constellation is considered and is partitioned as follows, the first level divides that constellation into two subsets, each of four points, and each set of points are mapped to a zero and one. Next, each set is further divided into two other subsets to create a subset of two points. This process continues until the last level where each subset has a single constellation point. The demapping process takes place in a sequential manner where the highest outer level (most significant bit) is demapped first then the second higher outer level based on the previous decoded bit until it reaches the last level which is the least significant bit. Thus, at each level, the demapped bits are fed to the subsequent level until all sub-frames are decoded.

%Next, the numerical results are presented where the simulated average coding rate of the MLPCM scheme is evaluated and compared with the theoretical results provided in Sec.~\ref{sec:outage_capacity} and~\ref{sec:Constellation}.  
% polar decoder

%\begin{figure}
 %   \centering
  %  \includegraphics[width=0.5\linewidth]{Figures/MLM.pdf}
   % \caption{Multilevel 8-QAM constellation Diagram}
    %\label{fig:8QAM_ML}
%\end{figure}
 \begin{figure}[t]
     \begin{subfigure}[b]{0.48\textwidth}
         \centering
         \includegraphics[height=0.82\linewidth,width=\linewidth]{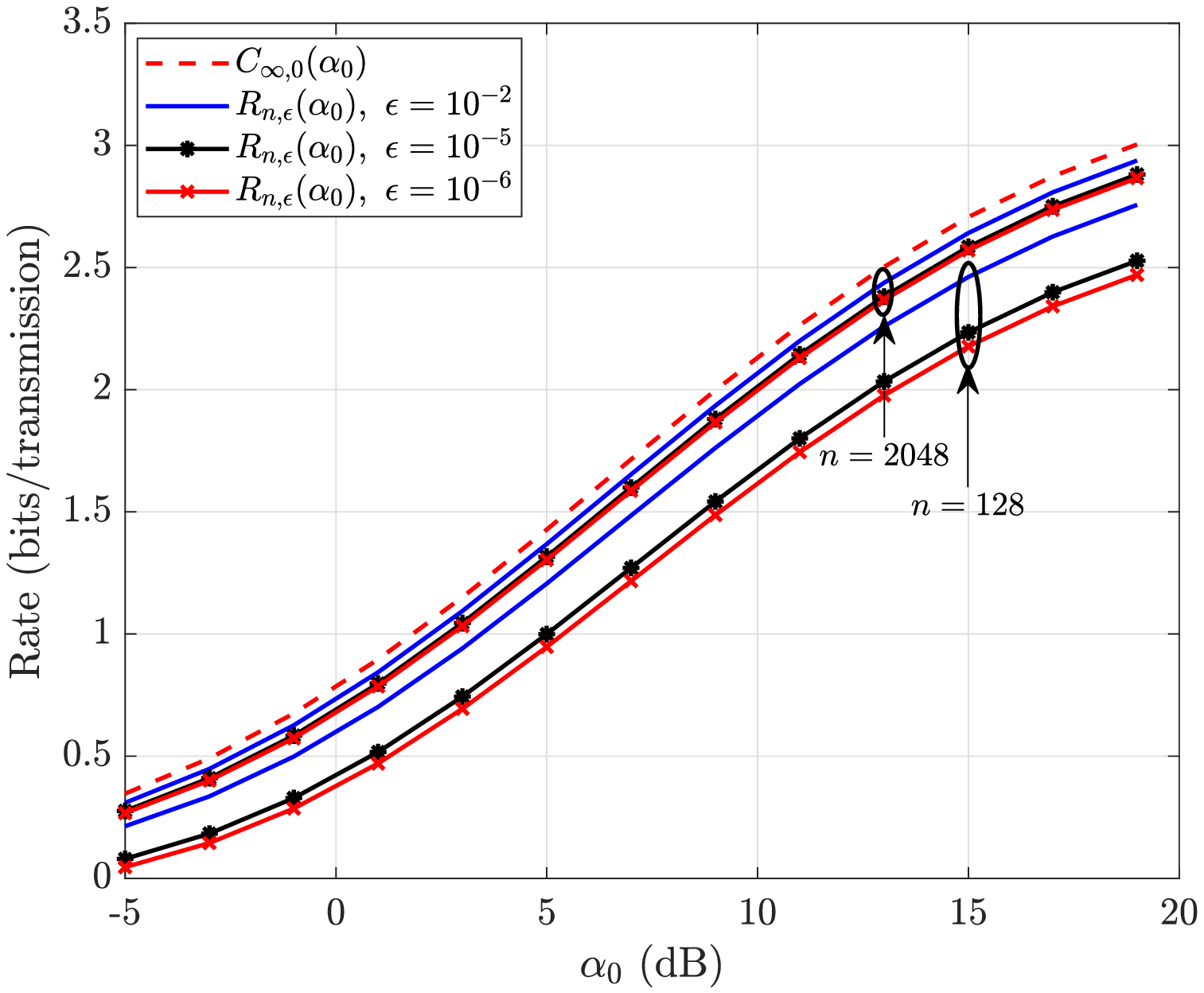}
         \caption{The maximum average coding rate at $r_0=250\ m$ for different values of $n$ and $\epsilon$}
         \label{fig:n}
     \end{subfigure}
     \begin{subfigure}[b]{0.48\textwidth}
         \centering
      \includegraphics[height=0.82\linewidth,width=1\linewidth]{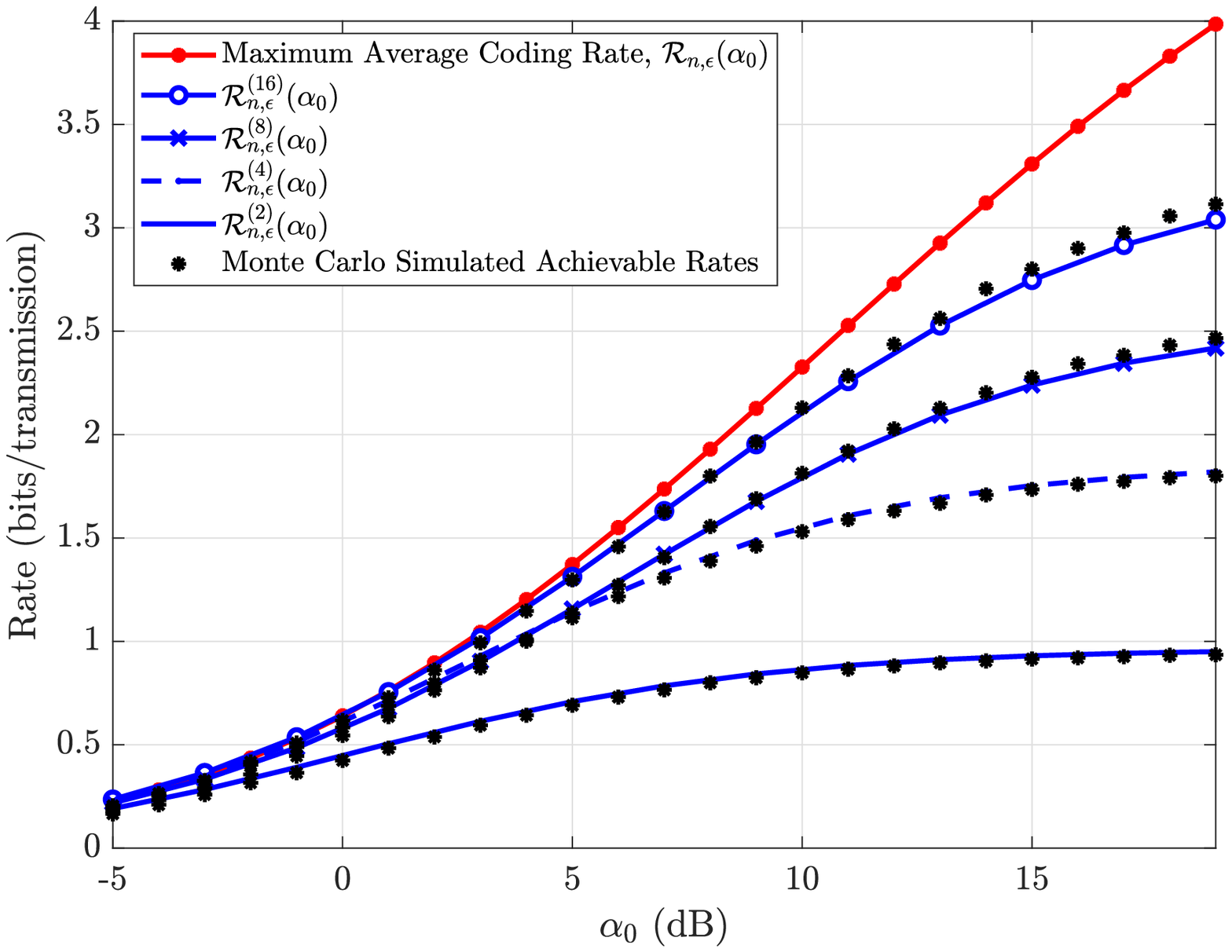}
         \caption{The average coding rate for M-QAM at $r_0=150\ m$, $n=128$, and $\epsilon=10^{-2}$.}
         \label{fig:Capacity_QAM}
     \end{subfigure}
     \caption{Average coding rate analysis of a large-scale network with $\lambda=1\ \mathrm{BS/km^2}$ for fixed $r_0$ }
     \label{fig:Average_Coding_Rate}
\end{figure}

%% file: MLPCM_figure.tex
\begin{figure}
\centering
\begin{tikzpicture}[thick,scale=0.7, every node/.style={scale=0.7}]
\draw [->,line width=1] (-3.5 ,0) -- node[pos=0.2,above,thick,font=\normalsize] {$\bold{X}$}(-2.5,0) ;
\node (demux) at (-2,0) [draw,thick,minimum width=1cm,minimum height=5cm, font=\normalsize] {\rotatebox{90}{Demultiplexer}} ;

\draw [->,line width=1] (-1.5,2) -- node[above,font=\scriptsize] {$x_{1}$}(0,2) ;
\draw [->,line width=1] (-1.5 ,1) -- node[above,font=\scriptsize] {$x_{2}$} (0,1);
\filldraw[color=white, fill=black,  thick](-1,0.5) circle (0.08);
\filldraw[color=white, fill=black,  thick](-1,0) circle (0.08);
\filldraw[color=white, fill=black,  thick](-1,-0.5) circle (0.08);
\draw [->,line width=1] (-1.5 ,-1) -- (0,-1);
\draw [->,line width=1] (-1.5 ,-2) -- node[above,font=\scriptsize] {$x_{\log_2 M}$}  (0,-2) ;

\node (Encoder1) at (0.75,2) [draw,thick,minimum width=1cm,minimum height=0.7cm, font=\small] {Encoder} ;

\node (Encoder2) at (0.75,1) [draw,thick,minimum width=1cm,minimum height=0.7cm, font=\small] {Encoder} ;

\node (Encoder3) at (0.75,-1) [draw,thick,minimum width=1cm,minimum height=0.7cm, font=\small] {Encoder} ;

\node (Encoder4) at (0.75,-2) [draw,thick,minimum width=1cm,minimum height=0.7cm, font=\small] {Encoder} ;

\node (demod) at (3,0) [draw,thick,minimum width=1cm,minimum height=5cm, font=\normalsize] {\rotatebox{90}{M-ary  Modulator}} ;

\draw [->,line width=1] (1.5,2) -- node[above,font=\scriptsize] {$u_{1}$}(2.5,2) ;
\draw [->,line width=1] (1.5 ,1) -- node[above,font=\scriptsize] {$u_{2}$} (2.5,1);
\filldraw[color=white, fill=black,  thick](2,0.5) circle (0.08);
\filldraw[color=white, fill=black,  thick](2,0) circle (0.08);
\filldraw[color=white, fill=black,  thick](2,-0.5) circle (0.08);
\draw [->,line width=1] (1.5 ,-1) -- (2.5,-1);
\draw [->,line width=1] (1.5 ,-2) -- node[above,font=\scriptsize] {$u_{\log_2 M}$} (2.5,-2) ;

\draw [->,line width=1] (3.5 ,0) --  node[pos=0.28,above,font=\normalsize] {${s_0}$} (4.5,0);

\node (channel) at (5.32,0) [draw,thick,align=center,minimum width=1cm,minimum height=1cm, font=\normalsize] {Channel} ;
\draw [->,line width=1] (6.18 ,0) --  node[pos=0.65,above,font=\normalsize] {$\bold{y}$}   (7,0);

\node (transmitter) at (0.5,0) [draw,dashed,thick,align=center,minimum width=7cm,minimum height=6cm, font=\large] {} ;

%% receiver
\draw [-,line width=1] (7 ,2.2) --  (7,-2.2);

\draw [->,line width=1] (7 ,2.2) --  (7.4,2.2);
\node (decoder1) at (8.15,2.2) [draw,thick,align=center,minimum width=0.7 cm,minimum height=0.7cm, font=\small] {Decoder} ;
\draw [->,line width=1] (8.9 ,2.2) -- node[pos=0.8,above,font=\scriptsize] {$\widetilde{x}_1$} (11.5,2.2);
% Decoder 2
\draw [->,line width=1] (7 ,1) --  (7.4,1);
\node (Delay2) at (7.95,1) [draw,thick,align=center,minimum width=0.7cm,minimum height=0.5cm, font=\small] {Delay} ;
\node (Decoder2) at (9.6,1) [draw,thick,align=center,minimum width=0.5cm,minimum height=0.7cm, font=\small] {Decoder} ;
\draw [->,line width=1] (8.55 ,1) --  (8.85,1);
\draw [->,line width=1] (10.35 ,1) -- node[pos=0.7,above,font=\scriptsize] {$\widetilde{x}_2$} (11.5,1);
% vertical branch from decoder 1 to decoder 2
\draw [->,line width=1] (9.6 ,2.2) --  (9.6,1.35);
%  dots
\filldraw[color=white, fill=black,  thick](7.5,0.4) circle (0.08);
\filldraw[color=white, fill=black,  thick](7.5,0.05) circle (0.08);
\filldraw[color=white, fill=black,  thick](7.5,-0.3) circle (0.08);
% Decoder 3
\draw [->,line width=1] (7 ,-0.9) --  (7.4,-0.9);
\node (Delay2) at (7.95,-0.9) [draw,thick,align=center,minimum width=0.5cm,minimum height=0.7cm, font=\small] {Delay} ;
\node (Decoder2) at (9.7,-0.9) [draw,thick,align=center,minimum width=0.5cm,minimum height=0.7cm, font=\small] {Decoder} ;
\draw [->,line width=1] (8.5 ,-0.9) --  (8.95,-0.9);
\draw [->,line width=1] (10.44 ,-0.9) -- node[pos=0.5,above,font=\tiny] {$\widetilde{x}_{\log_2\hspace{-0.09cm} M\hspace{-0.09cm}-\hspace{-0.05cm}1}$} (11.5,-0.9);
% vertical line
\draw [->,line width=1] (10 ,0) --  node[pos=0.2,above,font=\scriptsize] {$\widetilde{x}_{\log_2\hspace{-0.09cm}M\hspace{-0.05cm}-\hspace{-0.05cm}2}$} (10,-0.55);
\draw [->,line width=1] (9.2 ,0) --  node[pos=0.2,above,font=\scriptsize] {$\widetilde{x}_{1}$} (9.2,-0.55);
\filldraw[color=white, fill=black,  thick](9.4 ,-0.3) circle (0.07);
\filldraw[color=white, fill=black,  thick](9.6 ,-0.3) circle (0.07);
\filldraw[color=white, fill=black,  thick](9.8 ,-0.3) circle (0.07);
% Decoder 4
\draw [->,line width=1] (7 ,-2.2) --  (7.4,-2.2);

\node (Delay2) at (7.95,-2.2) [draw,thick,align=center,minimum width=0.5cm,minimum height=0.7cm, font=\small] {Delay} ;
\node (Decoder2) at (9.7,-2.2) [draw,thick,align=center,minimum width=0.5cm,minimum height=0.7cm, font=\small] {Decoder} ;
\draw [->,line width=1] (8.5 ,-2.2) --  (8.95,-2.2);
\draw [->,line width=1] (10.44 ,-2.2) -- node[pos=0.5,above,font=\scriptsize] {$\widetilde{x}_{\log_2\hspace{-0.09cm} M}$} (11.5,-2.2);
% vertical line
\draw [->,line width=1] (10 ,-1.5) --  (10,-1.85);
\draw [-,line width=1] (10 ,-1.5) --  (10.85,-1.5);
\draw [-,line width=1] (10.85 ,-1.5) --  (10.85,-0.9);
\draw [->,line width=1] (9.2 ,-1.5) -- node[pos=0.4,above,font=\scriptsize] {$\widetilde{x}_{1}$} (9.2,-1.85);
\filldraw[color=white, fill=black,  thick](9.4 ,-1.6) circle (0.07);
\filldraw[color=white, fill=black,  thick](9.6 ,-1.6) circle (0.07);
\filldraw[color=white, fill=black,  thick](9.8 ,-1.6) circle (0.07);

\node (mux) at (12,0) [draw,thick,minimum width=1cm,minimum height=5cm, font=\normalsize] {\rotatebox{90}{Multiplexer}} ;

\node (Receiver) at (9.75,0) [draw,dashed,thick,align=center,minimum width=6.5cm,minimum height=6cm, font=\large] {} ;

\draw [->,line width=1.2] (12.5 ,0) -- node[pos=0.8,above,thick,font=\normalsize] {$\widetilde{\mathbf{X}}$} (13.5,0) ;
\end{tikzpicture}
\caption{Multilevel Coded Modulation Block Diagram}
\label{fig:MCM}
\end{figure}

%% file: Rate_Numerical_Results_Modified.tex
\subsection{Average Rate Numerical Results}\label{sec:rate_results}

{This subsection presents illustrative numerical results for the theoretical rate analysis provided in Sec.~\ref{sec:outage_capacity} and~\ref{sec:Constellation}, which are also validated via Monte Carlo simulation. The results presented investigate the effect of the blocklength $n$ and the FER $\epsilon$, the tightness of the proposed theoretical approximations, and the performance comparison between the FBR and the AR. Then, the performance of the MLPCM is investigated in comparison with the theoretical average coding rate. }

{Fig.~\ref{fig:Average_Coding_Rate} plots the maximum average coding rate and the average coding rate for M-QAM versus the average SNR {$\alpha_0$} for a fixed $r_0$. In Fig.~\ref{fig:n}, the average coding rate is plotted for $n\hspace{-0.05cm}\in\hspace{-0.05cm}\{128,2048\}$, FER $\epsilon\hspace{-0.05cm}\in\hspace{-0.05cm}\{10^{-2},\ 10^{-5},\ 10^{-6}\}$, and $r_0=250$ m. Given a FER $\epsilon$, a gap exists between the average coding rate in the FBR and the AR. {The gap is around $1$ dB for $n=2048$ and is between $3$ and $5$ dB for $n\hspace{-0.1cm}=\hspace{-0.1cm}128$ at $\epsilon\hspace{-0.1cm}=\hspace{-0.1cm}10^{-5}$ and $10^{-6}$, respectively, which shows the severity of the SNR penalty in the FBR compared to the AR.} The figure shows the trade-off between reliability (represented by FER) and maximum average coding rate at a given $n$, {where} the maximum average coding rate decreases/increases as the FER decreases/increases. {For instance, to achieve a given maximum average coding rate at $n\hspace{-0.1cm}=\hspace{-0.1cm}128$ and FER $\epsilon\hspace{-0.1cm}=\hspace{-0.1cm}10^{-5}$, we need around $2$ dB more power compared with the same blocklength $n$ under the FER $\epsilon=10^{-2}$.} Furthermore, the impact of {small} $n$ becomes more {severe when the} FER $\epsilon$ is smaller, {resulting in rate degradation. Overall, this shows that} it is impractical to use results from the AR in the FBR.} 

{Second, the maximum achievable coding rate for different modulation schemes, presented in Sec.~\ref{sec:Constellation}, is simulated for $r_0=150\ m$, $\lambda=1\ \text{BS}/ {\text{km}}^2$, blocklength $n=128$, and FER $\epsilon=10^{-2}$ {in Fig. \ref{fig:Capacity_QAM}}. Fig.~\ref{fig:Capacity_QAM} demonstrates the approximated maximum achievable coding rate using the Gamma approximation for various modulation orders, ${M}=\{2,4,8,16\}$,  {(Theorem~\ref{Theorem:ModulationCapacity}), along with the maximum achievable rate under Gaussian codebooks (Theorem \ref{Theorem:Avg_coding_Network})} versus $\alpha_0$. The Monte Carlo simulated maximum achievable coding rate for QAM is provided to show the tightness of the approximation. The approximated maximum achievable rate is very close to the simulated maximum achievable rate. Thus, the Gamma approximation in \eqref{eq:Gamma_app} {leads to a fairly tight approximation}. It is also observed that as the QAM modulation order increases, the performance improves towards the maximum average coding rate {(without Gaussian codebooks)} in a large-scale network {in the FBR (Theorem \ref{Theorem:Avg_coding_Network}).}}

{The spatially averaged coding rate is depicted in Fig~\ref{fig:Capacity_r0} versus the transmit signal to noise power ratio $\frac{\mathcal{P}}{\sigma_w^2}$.\footnote{{The figure is plotted for a wide range of $\frac{\mathcal{P}}{\sigma_w^2}$ for illustrative purposes, bearing in mind that some values of $\frac{\mathcal{P}}{\sigma_w^2}$ in the plot may be impractical.}} The figure shows that for high SNR (above $0$ dB), one cannot achieve the same rate and FER achieved at a given $n$ by decreasing $n$ and paying an SNR penalty. For instance, the rate achieved at $n=2048$, $\epsilon=10^{-5}$, $\text{SNR} = 10$ dB cannot be achieved at $n=128$ and $\epsilon=10^{-5}$. The only way to achieve such a rate at $n=128$ is by paying a FER penalty, which is impractical for applications that need high reliability. {Hence, it is important to select $n$} that balances the trade-off between reliability and coding rate. Similar to Fig.~\ref{fig:n}, Fig. \ref{fig:Capacity_r0} shows the importance of the results in this paper for characterizing {and optimizing} the performance of large-scale networks with latency-sensitive applications. }
 
  \begin{figure}[t]
     \begin{subfigure}[b]{0.5\textwidth}
     
         \centering
            \includegraphics[height=0.82\linewidth,width=\linewidth]{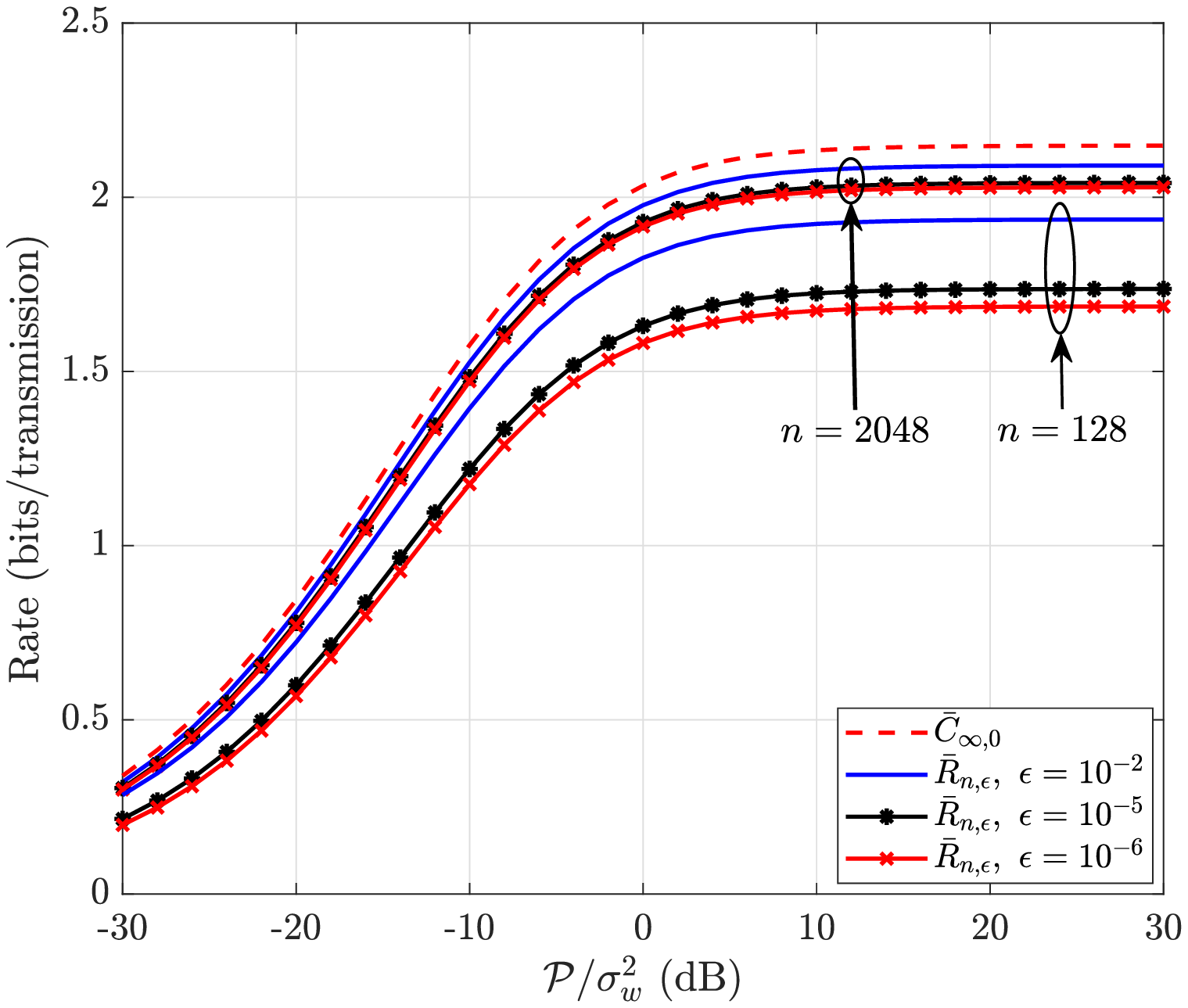}
         \caption{Maximum Average Coding Rate for  different values of $n$ and $\epsilon$. }
         \label{fig:Capacity_r0}
     \end{subfigure}
     \begin{subfigure}[b]{0.5\textwidth}
         \centering
         \includegraphics[height=0.8\linewidth,width=1\linewidth]{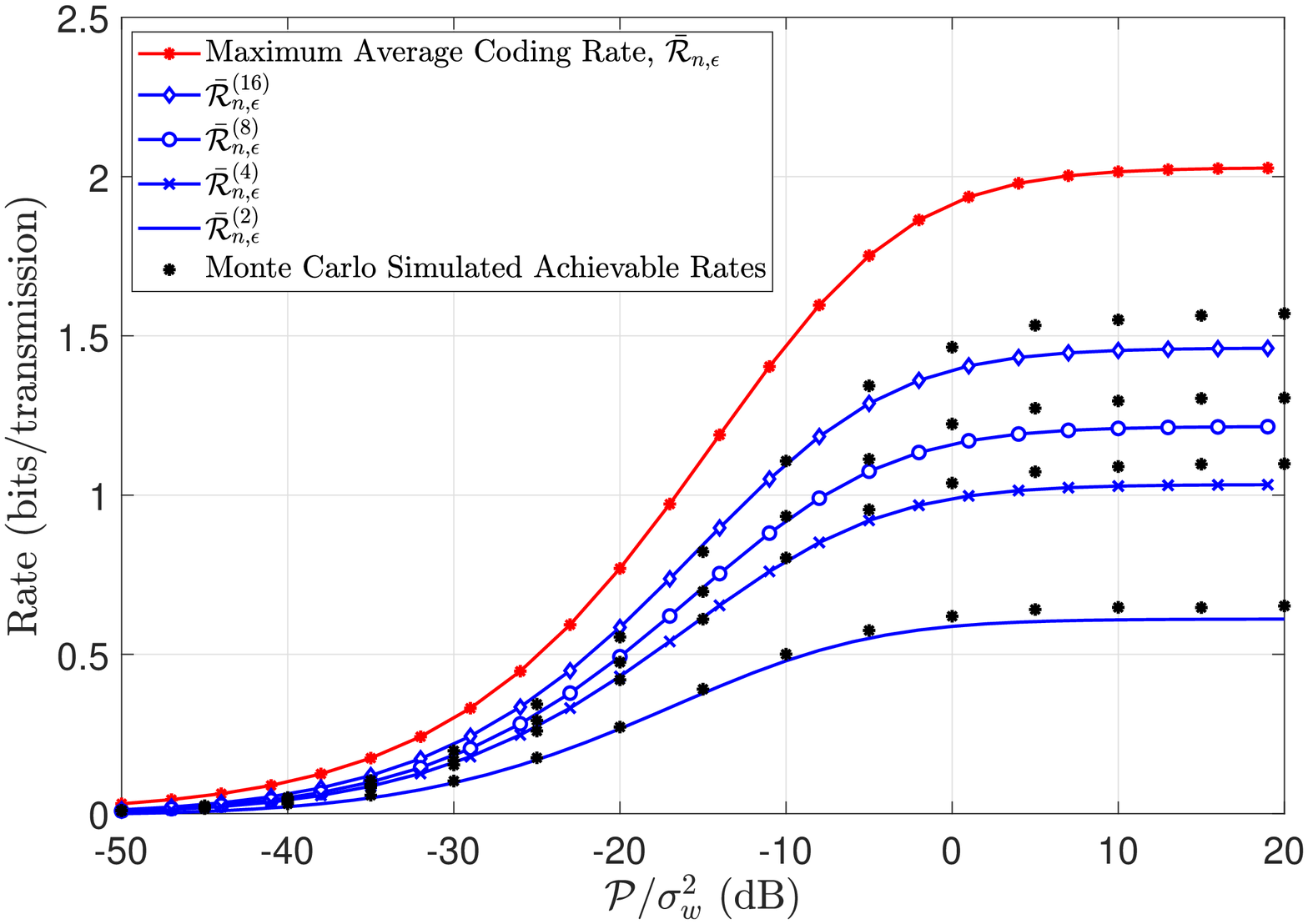}
         \caption{Average coding rate for M-QAM at $n=512$, and $\epsilon=10^{-2}$.}
         \label{fig:Modulation_r0}
     \end{subfigure}
     \caption{Average coding rate analysis for random $r_0$ at $\lambda=1\ \mathrm{BS/km^2}$ }
     \label{fig:Random_r0}
\end{figure}
{Fig.~\ref{fig:Modulation_r0} plots, the average achievable rate versus $\frac{\mathcal{P}}{\sigma_w^2}$ under a {Rayleigh distributed $r_0$}, { {QAM constellation with} $M=2,4,8,$ and $16$}, blocklenth $n=512$, and FER $\epsilon=10^{-2}${, showing a similar trend as} {for fixed $r_0$}. It is observed that the average achievable rate saturates at a rate of $\approx 2\ \text{bits/transmission}$ {for this choice of $\lambda$, $n$, and $\epsilon$, which is low compared to the case $r_0=150\ m$}. This performance is a consequence of averaging the performance of {nearby} UEs which have relatively strong received {desired signals} (high SINR) {and far away} UEs which suffer from more interference (low SINR).}

{{Finally,} considering the same setting in Fig.~\ref{fig:Capacity_QAM}, we simulate MLPCM with a QAM constellation for modulation orders $M=2,4,8,$ and $16$ in Fig.~\ref{fig:BPSK_QPSK_8QAM_16QAM}. {The simulations are performed by fixing the SNR and {increasing the coding rate until reaching} the FER $\epsilon$.} The achievable rates are plotted along with the theoretical expressions proposed in \eqref{eq:Capacity_FB_r_0}~and~\eqref{eq:modulation}. The figure shows that the achievable rates of the MLPCM scheme {with 2-QAM or 4-QAM approaches the theoretical rate of the corresponding QAM constellation.} However, for 8-QAM and 16-QAM, a gap of $\approx 2.5$ dB exists between the theoretical results and the simulation results{. Nonetheless, this SNR gap translates to a small rate gap of fractions of bits/transmission due to the slow increase of the rate versus SNR in the large-scale network which is caused by interference.} Fig.~\ref{fig:BPSK_QPSK_8QAM_16QAM} {thus validates that the theoretical performance under finite constellations can be approached using MLPCM, which validates the potential of MLPCM} for future technologies.}
% Polar Codes
\begin{figure}[t]
    \centering
    \includegraphics[height=0.39\linewidth,width=0.52\linewidth]{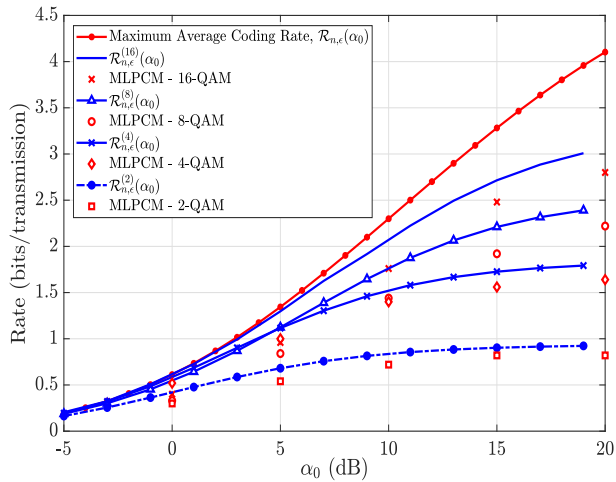}
    \caption{The average coding rate versus $r_0^{-\eta}\alpha$ at $\lambda=1\ \mathrm{BS/km^2}$, $r_0=150\ m$, and $n=128$.}
    \label{fig:BPSK_QPSK_8QAM_16QAM}
\end{figure}

%{Next, rate outage probability, reliability, and meta-distribution analysis are provided.}

%% file: Outage_Analysis.tex
\section{Rate Outage Probability \& Meta Distribution}\label{sec:Outage_Analysis}

In the previous section, we studied the case where the SINR is known at the transmitter. However, if the BS does not track the {SINR}, then the rate adaptation investigated in the previous section is not possible. 
{Considering the second mode of operation, where the SINR is unknown at the transmitter but known at the receiver, the transmitter encodes at a constant (target) rate $R_t$ and blocklength $n$. As a result, the resulting (conditional) FER conditioned on SINR $\Upsilon$ may or may not satisfy a desired design FER $\epsilon$ due to the unknown stochastic variation of the intended channel fading and aggregate interference. In this case, it is important to keep this conditional FER below a threshold. This is especially relevant bearing in mind that $n=L/l$, i.e., a large conditional FER means that there are coherence intervals that produce large FERs for all blocks within these coherence intervals. This can lead to bursts of frame errors when the channel undergoes a period of coherence intervals with weak channels, which is undesired.}

{Thus, we define an FER threshold $\bar{\epsilon}$, and it is desired to maintain the conditional FER below $\bar{\epsilon}$. If the conditional FER exceeds $\bar{\epsilon}$, we say that we have an outage. This  equivalently means that rate $R_{n,\bar{\epsilon}}(\Upsilon)$ supported by the channel at the given $n$ and FER threshold $\bar{\epsilon}$ decreases below $R_t$. This is defined formally as follows.}
%\footnote{This is in contrast with the AR where the ET is not taken into consideration  in outage since it is assumed that the FER is vanishingly small.}. In this case, we characterize performance by the outage probability define next.
\vspace{-0.4cm}
{\begin{definition}\label{def:outage}
Given a target rate $R_t$ and a DL large-scale network with a distance $r_0$ between the desired UE and the serving BS as defined in (1), we define the outage probability in the FBR with blocklength $n$ and FER threshold $\bar{\epsilon}$ as $\mathcal{O}(r_0,R_t  ,n,\bar{\epsilon}) = P(R_{n,\bar{\epsilon}} (\Upsilon) < R_t)$, where $R_{n,\bar{\epsilon}} (\Upsilon)$ is as defined in Lemma \ref{lem:poly}.
 \end{definition}}
This outage probability is studied in this section. 
%This section investigates the fixed rate transmission in the FBR. Such co  previous section considers a scenario where the transmitter adapts its rate to the channel variations. In some practical applications, the knowledge of CSI along with the instantaneous interference may not be feasible, then, it is desired to transmit at a fixed (target) rate which incurs some outage. The outage probability and reliability of this scenario are studied first, followed by the meta distribution of the coding rate to provide insights into the percentage of users achieving a certain performance. Finally, numerical results is presented and discussed.
\input{Outage}
\input{MetaDistribution}
\input{Outage_Numerical_Results}
% 

%% file: Outage.tex
\subsection{DL Outage Analysis in the Finite Blocklength Regime}\label{sec:Outage_DL}

We start by reviewing the outage probability of a large-scale network in the AR. In this case, the outage is defined as the probability that the channel capacity (defined as the maximum rate such that the error probability vanishes as $n\to\infty$) is lower than a {target rate $R_t$}. The rate outage probability under a given $r_0$ is given as~\cite{Andrew}
\begin{align}
\label{eq:outage}
 \mathcal{O}(r_0,R_t)=1-e^{-\frac{(2^{R_t}-1) \sigma_w^2 }{\mathcal{P} r_0^{-\eta}}} \mathcal{L}_{\mathcal{B}_g}\left\{\frac{2^{R_t}-1 }{\mathcal{P} r_0^{-\eta}}\right\}.
\end{align}

%\vspace{-0.5cm}
{Studying the rate outage probability under blocklength $n$ and FER threshold $\bar{\epsilon}$ is generally difficult to simplify as it is a function of the fading distribution as well as the network geometry. Instead, we provide bounds which are fairly tight in the following theorem.}
%\vspace{-0.3cm}
\begin{theorem}\label{Theorem:BoundsOutage_Network_fixed_r0}
For {an average SNR $\alpha_0$}, the rate outage probability defined in Def. \ref{def:outage} satisfies {$\mathcal{O}_l(r_0,R_t,n,\bar{\epsilon})\leq \mathcal{O}(r_0,R_t,n,\bar{\epsilon}) \leq \mathcal{O}_u(r_0,R_t,n,\bar{\epsilon})$}, where
%\vspace{-0.22cm}
\begin{align}
\mathcal{O}_l(r_0,R_t,n,\bar{\epsilon}) &=  \mathcal{O}(r_0,R_t), \label{eq:lower_bound}\\
\mathcal{O}_u(r_0,R_t,n,\bar{\epsilon})&= \mathcal{O}(r_0,R_t+a_{n,\bar{\epsilon}}{-b_n}),\label{eq:upper_bound}
\end{align}
where $\mathcal{O}(r_0,R_t)$ is defined in \eqref{eq:outage}, $a_{n,\bar{\epsilon}}= \sqrt{\frac{\log_2^2(e)}{ n} } Q^{-1}(\bar{\epsilon})$, and ${b_n=\frac{1}{2n} \log_2 n}$.
\end{theorem}
%\vspace{-0.3cm}
\begin{proof}
From Def.~\ref{def:outage}, we have
\begin{align}\label{eq:outage_general}
   \mathcal{O}(r_0,R_t,n,\bar{\epsilon})   =\mathbb{P}(R_{n,\bar{\epsilon}}(\Upsilon)<R_t)=\mathbb{P}\Bigg(\log_2 (1+\Upsilon)- {a_{n,\bar{\epsilon}}}\sqrt{1-\frac{1}{(1+\Upsilon)^2}}{+b_n}<R_t\Bigg).
   \end{align}
{Noting that $\mathcal{K}(\Upsilon)=\sqrt{1-\frac{1}{(1+\Upsilon)^2}}$ is monotonically increasing in $\Upsilon$ as shown in Fig. \ref{fig:K_SINR}, {satisfying $\mathcal{K}(\Upsilon)\in[0,1]$,} we conclude that
%\vspace{-0.2cm}
%\begin{align}
$\mathcal{O}_{l}(r_0,R_t,n,\bar{\epsilon})\leq \mathcal{O}(r_0,R_t,n,\bar{\epsilon}) \leq \mathcal{O}_{u}(r_0,R_t,n,\bar{\epsilon})$,
%\end{align}
where the lower {and upper} bounds are obtained by setting $\mathcal{K}(\Upsilon)$ to $0$ and $1$, respectively.}
\end{proof}
\begin{figure}
    \centering
    \includegraphics[width=0.5\linewidth]{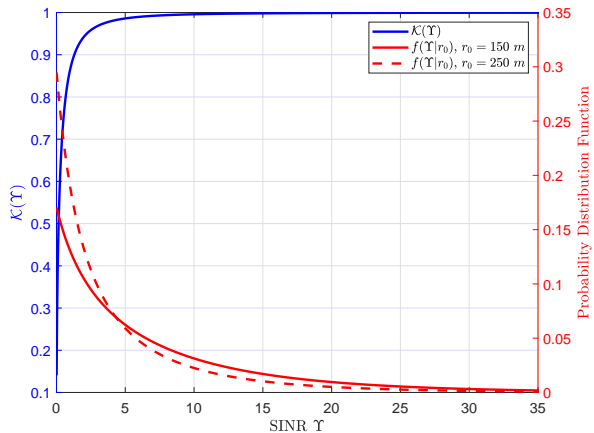}
    \caption{$\mathcal{K}(\Upsilon)$ versus $\Upsilon$ and the PDF of the SINR at $r_0\in\{150,250\}\ m$}
    \label{fig:K_SINR}
\end{figure}
{The outage probability lower bound confirms that the outage probability in the FBR is larger than that in the AR and the upper bound provides a guaranteed performance. By observing Fig.~\ref{fig:K_SINR}, it can be seen {that} $\mathcal{K}(\Upsilon)$ quickly approaches one as $\Upsilon$ grows. The figure also shows that the density $f(\Upsilon|r_0)$ in the interval of SINR, where the approximation $\mathcal{K}(\Upsilon)\approx1$ is tight ($\Upsilon > 5$ e.g. ), increases as $r_0$ decreases (which is the case for dense networks). Hence, the upper bound  $\mathcal{O}_{u}(r_0,R_t,n,\bar{\epsilon})$ can be used to evaluate guaranteed outage performance in general and can be used as a tight approximation for dense networks. {This is verified in Fig.~\ref{fig:Outage_App_r0} \& \ref{fig:Outage_App}}. Next, the spatially averaged outage probability is presented.}

\begin{Corollary}\label{Corollary:Outage_Probability_Random_r0}
The average rate outage probability defined in Def. \ref{def:outage} with blocklength $n$, FER $\bar{\epsilon}$, and target rate $R_t$ can be upper bounded as
%\vspace{-0.2cm}
\begin{align}\label{eq:Outage_r0_FB}
  \bar{\mathcal{O}}(R_t,n,\bar{\epsilon}) =  1-\int_{0}^{\infty} 2\pi \lambda r_0 &e^{\hspace{-.1cm}-\frac{2^{R_t+a_{n,\bar{\epsilon}}{-b_n}}-1 }{\alpha_0}-\pi \lambda r_0^2}\mathcal{L}_{\mathcal{B}_g}\left\{\frac{2^{R_t+a_{n,\bar{\epsilon}}{-b_n}}-1 }{\mathcal{P} r_0^{-\eta}}\right\}d r_0
\end{align}
\end{Corollary}
\begin{proof}
Using the law of iterated expectation, {the average outage probability is the average of} \eqref{eq:upper_bound} with respect to $r_0$, i.e., $\mathbb{P}(R>R_t)=\mathbb{E}_{r_0}\{\mathbb{P}(R>R_t|r_0)\}$. The result then follows since $r_0$ follows a Rayleigh distribution.
\end{proof}

For a shadowed urban cellular network, the path-loss exponent $\eta$ lies in the range $3-5$ \cite{eta}, where $\eta=4$ is a common value of practical relevance that is widely utilized in the literature. For the case $\eta=4$, the outage probability in~\eqref{eq:Outage_r0_FB} can be simplified as given next.
%\vspace{-0.3cm}
\begin{Corollary}
The average rate outage probability of the large-scale network modeled by \eqref{eq:system_model} with blocklength $n$, FER $\bar{\epsilon}$, target rate $R_t$, and path loss exponent $\eta=4$, can be upper bounded as
%\vspace{-0.25cm}
\begin{equation}
\bar{ \mathcal{O}}(R_t,n,\bar{\epsilon}) =  1-\pi \lambda  \sqrt{\frac{\pi \mathcal{P}  }{\sigma_w^2 \mu}} e^{\frac{ \sigma_w^2 \mu z^2}{4  \mathcal{P}} }Q\left(\sqrt{\frac{z^2 \mu \sigma_w^2}{2  \mathcal{P}}}\right)
\end{equation}
where $\mu = 2^{R_t+a_{n,\bar{\epsilon}}{-b_n}}-1$  and $ z=  \frac{\lambda \pi  \mathcal{P} }{\sigma_w^2 \mu} \left(\sqrt{\mu} \arctan(\sqrt{\mu})+1\right)$.
\end{Corollary}
%\vspace{-0.25cm}
\begin{IEEEproof}%\renewcommand\qedsymbol{}
A closed-form expression for the upper bounded outage probability at $\eta=4$ is obtained as follows
%\vspace{-0.5cm}
\begin{align}
    \bar{ \mathcal{O}}(R_t,n,\bar{\epsilon})&= 1-\int_{0}^{\infty} e^{-\frac{\mu r_0^{4}   \sigma_w^2}{\mathcal{P}}}\mathcal{L}_{\mathcal{B}_g}\left\{\frac{\mu r_0^{4} }{\mathcal{P}}\right\}  f_{r_0}(r_0)\ d r_0 \\
   &\overset{(i)}{=} 1- \int_{0}^{\infty} 2 \pi \lambda r_0 e^{-\frac{  \sigma_w^2 \mu}{\mathcal{P} } \left(r_0^4 +z\ r_0^2\right)} d r_0\\
    %&=1- \pi \lambda e^{ \frac{z^2 \mu}{4\alpha}} \int_{0}^{\infty} 2 r_0 e^{-\frac{\mu}{\alpha} \left(r_0^2 +\frac{z}{2}\right)^2} d r_0\nonumber\\
    &\overset{(ii)}{=} 1-\pi \lambda \sqrt{\frac{\pi\ \mathcal{P} }{\sigma_w^2\mu}} e^{ \frac{z^2 \mu\sigma_w^2}{4 \mathcal{P} }}  Q\left(\sqrt{\frac{z^2 \mu\sigma_w^2}{2 \mathcal{P} }}\right)\label{eq:outage_without_r_0}
\end{align}
where we used ${\mathcal{L}_{\mathcal{B}_g}\{s\}}_{\eta=4}=\exp\left\{-\pi \lambda \sqrt{s\mathcal{P}} \text{tan}^{-1}\left(\sqrt{\frac{s \mathcal{P}}{r_0^4}} \right) \right\}$, $(i)$~is obtained {using a change of} variable
$z=\frac{\lambda \pi  \mathcal{P}  }{\mu\sigma_w^2}\left( \sqrt{\mu} \arctan(\sqrt{\mu})+1\right)$, and $(ii)$ is obtained {using the change of variables and} the CDF of the Gaussian distribution with mean $\left(\frac{z}{2}\right)$ and variance $\left(\frac{\mathcal{P} }{2\mu\sigma_w^2}\right)$.
\end{IEEEproof}

The rate outage probability of a large-scale network in the FBR is upper bounded, with guaranteed performance at high SINR, in a generic form in~\eqref{eq:upper_bound} and~\eqref{eq:Outage_r0_FB}  and in a closed-form expression for $\eta=4$ in~\eqref{eq:outage_without_r_0}, as a function of BS density $\lambda$, FER $\bar{\epsilon}$, blocklength $n$, and average SNR $\alpha_0$. Next, the reliability of the network under the FBR.

\subsection{Reliability of a Large-Scale Network in the Finite Blocklength Regime}
{We define a metric to evaluate the overall network reliability for a fixed $r_0$, defined as the probability of correctly decoding a codeword. A codeword is decoded correctly in two cases in the FBR. The first case is when $\mathcal{R}_{n,\bar{\epsilon}}(\Upsilon)>R_{th}$, in which case the codeword is decoded correctly with a probability larger than or equal to $(1-\bar{\epsilon})$. The second case is when $\mathcal{R}_{n,\bar{\epsilon}}(\Upsilon)<R_{th}$, in which case few codewords may still be decoded correctly. Since in practice, it is desirable to operate at low outage probability, hence the second case does not contribute much to reliability. Thus, to simplify the analysis, we will lower bound the probability of decoding a codeword correctly in the second case by $0$ (i.e., its contribution to reliability is neglected) and we will only consider the first case. Hence, the reliability is lower bounded by the following expression, providing a guaranteed performance for a given $r_0$,}
\begin{equation}
    \mathcal{T}_{n,\bar{\epsilon}}(r_0,R_t)=(1-\mathcal{O}(r_0,R_t,n,\bar{\epsilon}))(1-\bar{\epsilon}).
\end{equation}
The guaranteed reliability averaged over $r_0$ is given by
\begin{equation}
     \bar{\mathcal{T}}_{n,\bar{\epsilon}}(R_t)=(1-\bar{\mathcal{O}}(R_t,n,\bar{\epsilon}))(1-\bar{\epsilon}),
\end{equation}
Also, for the sake of comparison, the reliability of the AR can be defined as
\begin{equation}
     {\mathcal{T}}_{\infty,0}(r_0,R_t)=1-\mathcal{O}(r_0,R_t),
\end{equation}
 since the error probability is assumed to be vanishing as $n\to\infty$, and hence failure occurs when there is an outage. The reliability will be assessed numerically in Sec. \ref{sec:outage_results}.

Studying the rate outage probability and the overall system reliability gives more insight into the performance of the network. Moreover, studying the coding rate meta distribution, which provides the percentiles of users achieving a rate $R_t$, will provide a complete picture of the network performance. Hence, the meta distribution of the coding rate is derived next.

%% file: MetaDistribution.tex
\subsection{{Coding Rate Meta Distribution in the Finite Blocklength Regime}}\label{sec:Meta}
 The meta distribution provides fine-grained information about network performance, in the form of the fraction of users that can achieve a minimum data rate $R_t$ at a FER threshold $\bar{\epsilon}$ {with probability at least $p_{t}$}. {The definition and analysis of this quantity are provided next.} 
{\begin{definition}
Given a target rate $R_t$ and a DL large-scale network with a distance $r_0$ between the desired UE and the serving BS as defined in (1), we define the coding rate meta distribution in the FBR with blocklength $n$ and FER threshold $\bar{\epsilon}$ as $F_{R_t}(p_t) = {P}(P_s(R_t,n,\bar{\epsilon})<p_s)$ where $P_s(R_t,n,\bar{\epsilon})= P(R_{n,\bar{\epsilon}}(\Omega) > R_t|\Tilde{\Psi})$ is the probability of success, $\Omega$ is the signal-to-interference ratio (SIR) and $R_{n,\bar{\epsilon}} (\Omega)$ is as defined in Lemma \ref{lem:poly}.
 \end{definition}}

For simplicity, we assume an interference-limited scenario where noise is neglected and the SIR is given by {$\Omega=\frac{\mathcal{P} r_o^{-\eta}|h_0|^2}{\sum_{r_i\in \Tilde{\Psi}/r_o} \mathcal{P} r_i^{-\eta}|h_i|^2}$. } Then, $P_s(R_t,n,\bar{\epsilon})$ is given by
\begin{align}\label{eq:Ps_exact}
    P_s(R_t,n,\bar{\epsilon})=\mathbb{P}^{^^21}(R_{n,\bar{\epsilon}}(\Omega) > R_t|\Tilde{\Psi})=\mathbb{P}^{!}\left(\log_2(1+\Omega)-\frac{\sqrt{V(\Omega)} Q^{-1}(\bar{\epsilon})}{\sqrt{n}}{+b_n}> R_t|\Tilde{\Psi}\right)
\end{align}
where $\mathbb{P}^{^^21}(\cdot)$ is {the reduced Palm measure of the point process \cite{Meta1},\cite[Def. 8.8]{reducedPalm_book}}. Note that $P_s(R_t,n,\bar{\epsilon})$ is a random variable which is a function of the stochasticity of the network. Its moments are defined as
%\vspace{-0.4cm}
\begin{align}\label{eq:moments_exact}
    \mathcal{M}_d=\mathbb{E}\left\{P_s(R_t,n,\bar{\epsilon})^d\right\},
\end{align}
%\vspace{-0.25cm}
and are used to calculate the meta distribution as follows~\cite{Meta1}
\begin{align}\label{eq:meta_exact}
    F_{R_t}(p_{t})=\mathbb{P}(P_s(R_t,n,\bar{\epsilon})>p_{t}){=\frac{1}{2}+\frac{1}{\pi}\int_{0}^{\infty}\frac{\mathcal{I}m(e^{-t \log(p_{t})} \mathcal{M}_{jt})}{t} dt},
\end{align}
{where $\mathcal{I}m(\cdot)$ is the imaginary part of a complex number and $\mathcal{M}_{jt}$ is obtained at $d = j t$ in \eqref{eq:moments_exact}.}
%The expressions provided in \eqref{eq:Ps_exact},\eqref{eq:moments_exact}, and \eqref{eq:meta_exact} are evaluated using Monte-Carlo simulation.

{{In~\cite{Meta1}}, a simple yet tight approximating distribution {of $P_s(R_t,n,\bar{\epsilon})$} is proposed  for the AR, which is the beta distribution given by
\begin{align}
    f(P_s;R_t)\approx\frac{P_s^{\frac{\vartheta(\beta+1)-1}{1-\vartheta}}(1-P_s)^{\beta-1}}{\text{B}\left(\frac{\vartheta \beta}{1-\vartheta},\beta\right)},
\end{align}
where $\text{B}(\cdot,\cdot)$ is the beta function, $\vartheta=\mathcal{M}_1$ and the $\beta=\frac{(\mathcal{M}_1-\mathcal{M}_2)(1-\mathcal{M}_1)}{\mathcal{M}_2-\mathcal{M}_1^2}$. Hence, the meta distribution could be computed as the CCDF of the beta distribution. However, $\mathcal{M}_1$ and $\mathcal{M}_2$ provided in \eqref{eq:moments_exact} are difficult to evaluate in the FBR, and hence, an approximation is proposed in this paper. At high SIR, where the term $\mathcal{K}(\Upsilon){=\sqrt{1-\frac{1}{(1+\Upsilon)^2}}}$ is approximately {equal to} 1, the approximated probability of success is given by}
%\vspace{-0.3cm}

\begin{figure}[t]
     \begin{subfigure}[h]{0.5\textwidth}
         \centering
         \includegraphics[width=1\linewidth,height=6.2cm]{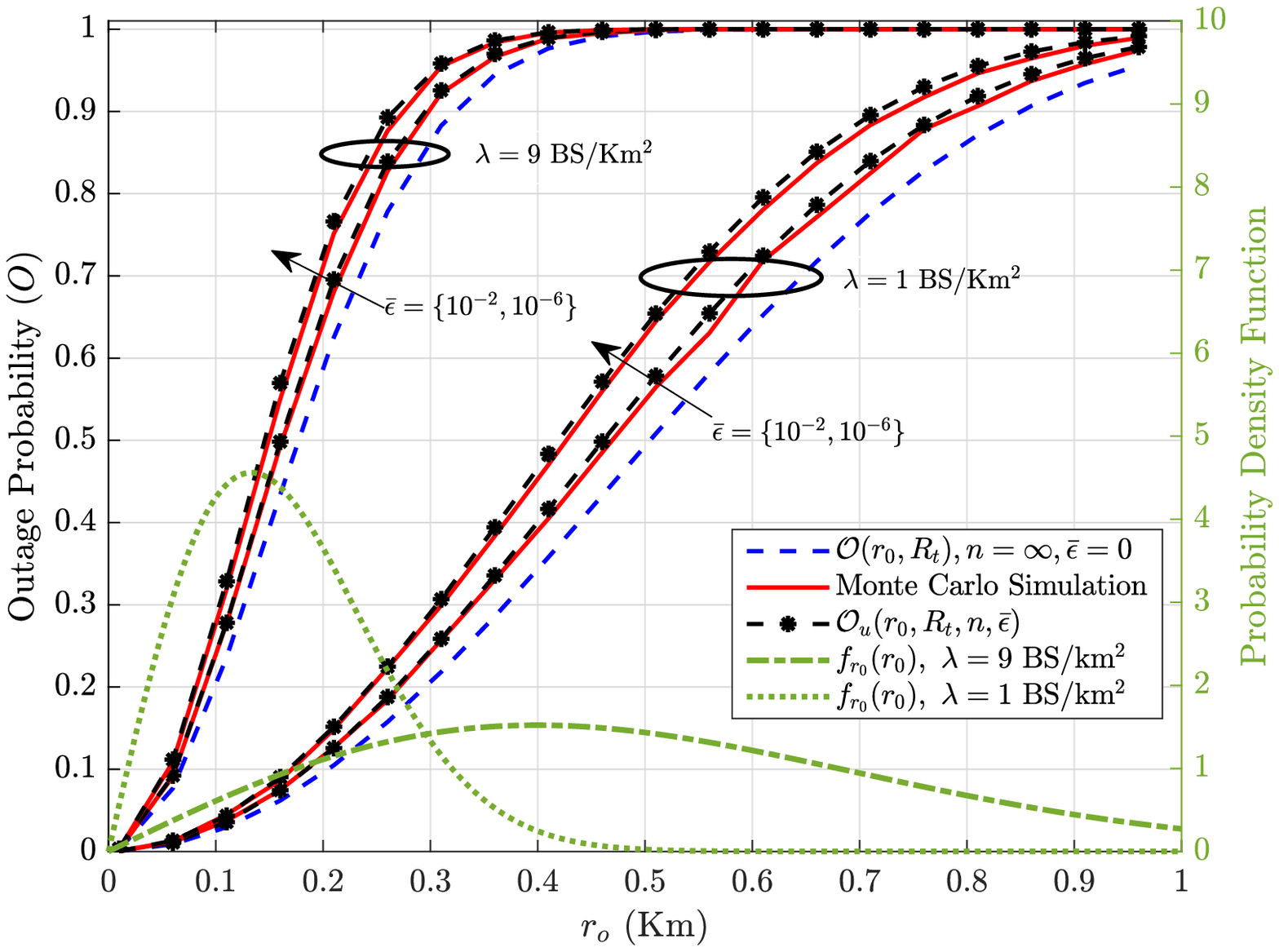}
         \caption{$n=128$}
         \label{fig:Outage_App_128}
     \end{subfigure}
     \begin{subfigure}[h]{0.5\textwidth}
         \centering
         \includegraphics[width=1\linewidth,height=6.2cm]{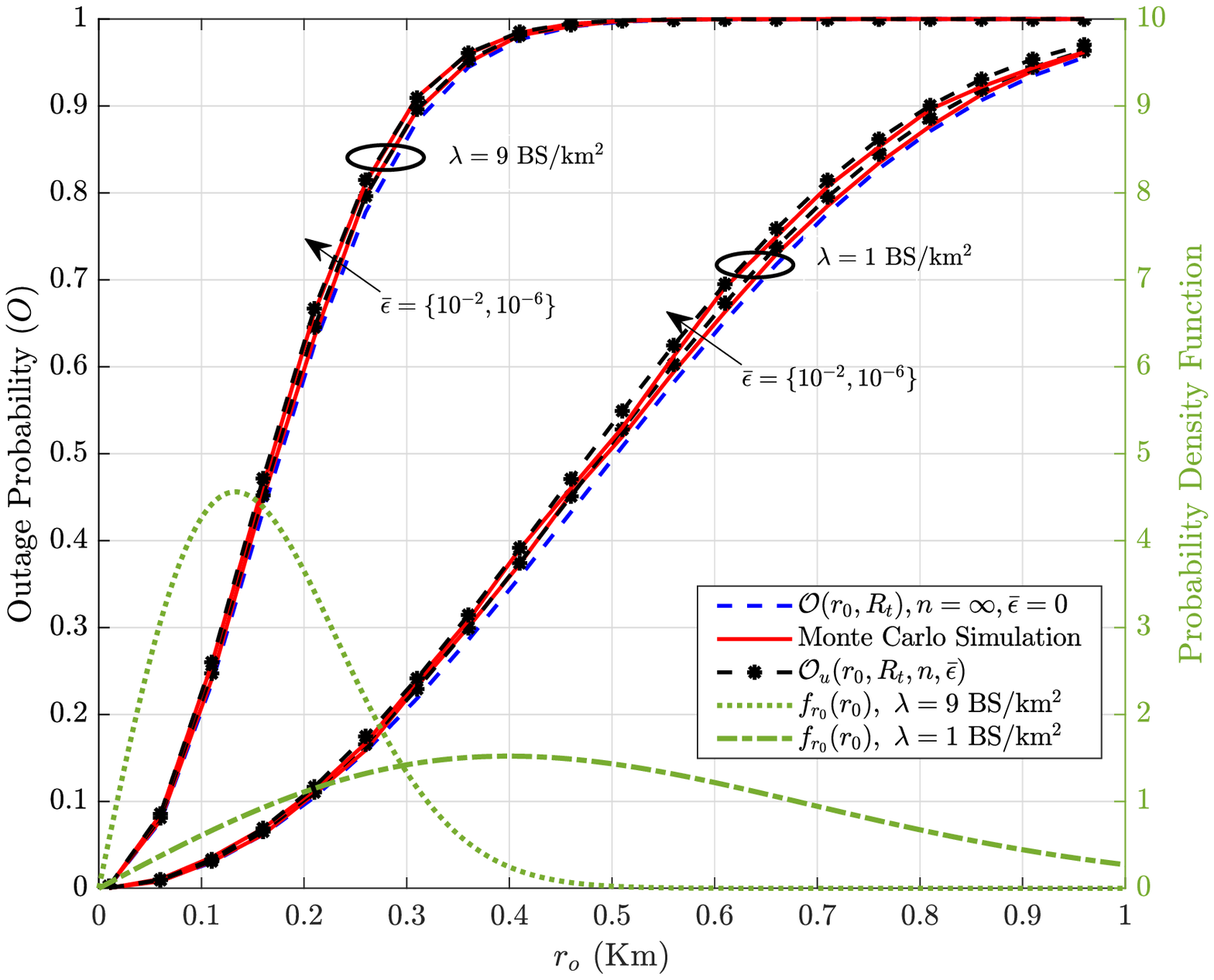}
         \caption{$n=2048$}
         \label{fig:Outage_App_2048}
     \end{subfigure}
     \caption{The outage probability versus $r_0$ at $\lambda\in\{1, 9\}\ \mathrm{BS/km^2}$, $\bar{\epsilon}\in\{10^{-2},10^{-6}\}$, and $\alpha=0\ \text{dB}$.}
     \label{fig:Outage_App_r0}
\end{figure}
\begin{align}
    \hat{P}_s(R_t,n,\bar{\epsilon})&=\mathbb{P}^{!}\left(\log_2 (1+\Omega)- a_{n,\bar{\epsilon}}{+b_n} > R_t|\Tilde{\Psi}\right). \nonumber
\end{align}
Therefore, the approximated probability of success {can be expressed as a function of} the SIR instead of the coding rate, as follows,
\begin{align}\label{eq:P_App}
    \hat{P}_s(R_t,n,\bar{\epsilon})=\mathbb{P}^{!}\left(\Omega > 2^{R_t+a_{n,\bar{\epsilon}}{-b_n}}-1|\Tilde{\Psi}\right) = \prod_{r_i\in \Tilde{\Psi}\setminus \{r_0\}} \frac{1}{1+(2^{R_t+a_{n,\bar{\epsilon}}{-b_n}}-1)\left(\frac{r_0}{r_i}\right)^{\eta}},
\end{align}
The expression provided in \eqref{eq:P_App} is obtained by averaging over the channel gains. Then, its moments are given by
\begin{align}\label{eq:approx_moments}
    \hat{\mathcal{M}}_d&=\mathbb{E}\left\{\hat{P}_s(R_t,n,\bar{\epsilon})^d\right\}\nonumber\\
    &= \mathbb{E}\left\{\prod_{r_i\in \Tilde{\Psi}\setminus \{r_0\}} \frac{1}{\left(1+(2^{R_t+a_{n,\bar{\epsilon}}{-b_n}}-1)\left(\frac{r_0}{r_i}\right)^{\eta}\right)^d}\right\}\nonumber\\
     &=\exp\left\{-2 \pi \lambda  \int_{r_0}^{\infty} r \left(1-\frac{1}{(1+(2^{R_t+a_{n,\bar{\epsilon}}{-b_n}}-1) (r_0/r)^{\eta})^d}\right) d r \right\},
\end{align}
{Thus, the meta distribution can be evaluated using the approximated moments in \eqref{eq:approx_moments} as follows}
\begin{align}
      \hat{F}_{R_t}(p_{t})=\mathbb{P}(\hat{P}_s(R_t,n,\bar{\epsilon})>p_{t})=\frac{1}{2}+\frac{1}{\pi}\int_{0}^{\infty}\frac{\mathcal{I}m(e^{-t \log(p_{t})} \hat{\mathcal{M}}_{jt})}{t} dt.
\end{align}
Using the beta distribution approximation, the meta-distribution can be approximated as 
\begin{align}\label{eq:meta_beta}
      \hat{F}_{R_t}(p_{t})\approx 1- I_{p_{t}}\left(\frac{\hat{\mathcal{M}}_1(\hat{\mathcal{M}}_1-\hat{\mathcal{M}}_2)}{(\hat{\mathcal{M}}_2-\hat{\mathcal{M}}_1^2)},\frac{(1-\hat{\mathcal{M}}_1)(\hat{\mathcal{M}}_1-\hat{\mathcal{M}}_2)}{(\hat{\mathcal{M}}_2-\hat{\mathcal{M}}_1^2)}  \right),
\end{align}
{where $I_{z}(x,y)=\frac{1}{\text{B}(x,y)}\int_{0}^{z} t^{x-1} (1-t)^{y-1}dt$ is the regularized incomplete beta function.}

%The main question next is: How do the expressions in this section compare with similar expression for the AR? This question is addressed in the following numerical evaluations.

%% file: Outage_Numerical_Results.tex
\begin{figure}[t]
     \begin{subfigure}[h]{0.5\textwidth}
         \centering
         \includegraphics[width=1\linewidth,height=6.4cm]{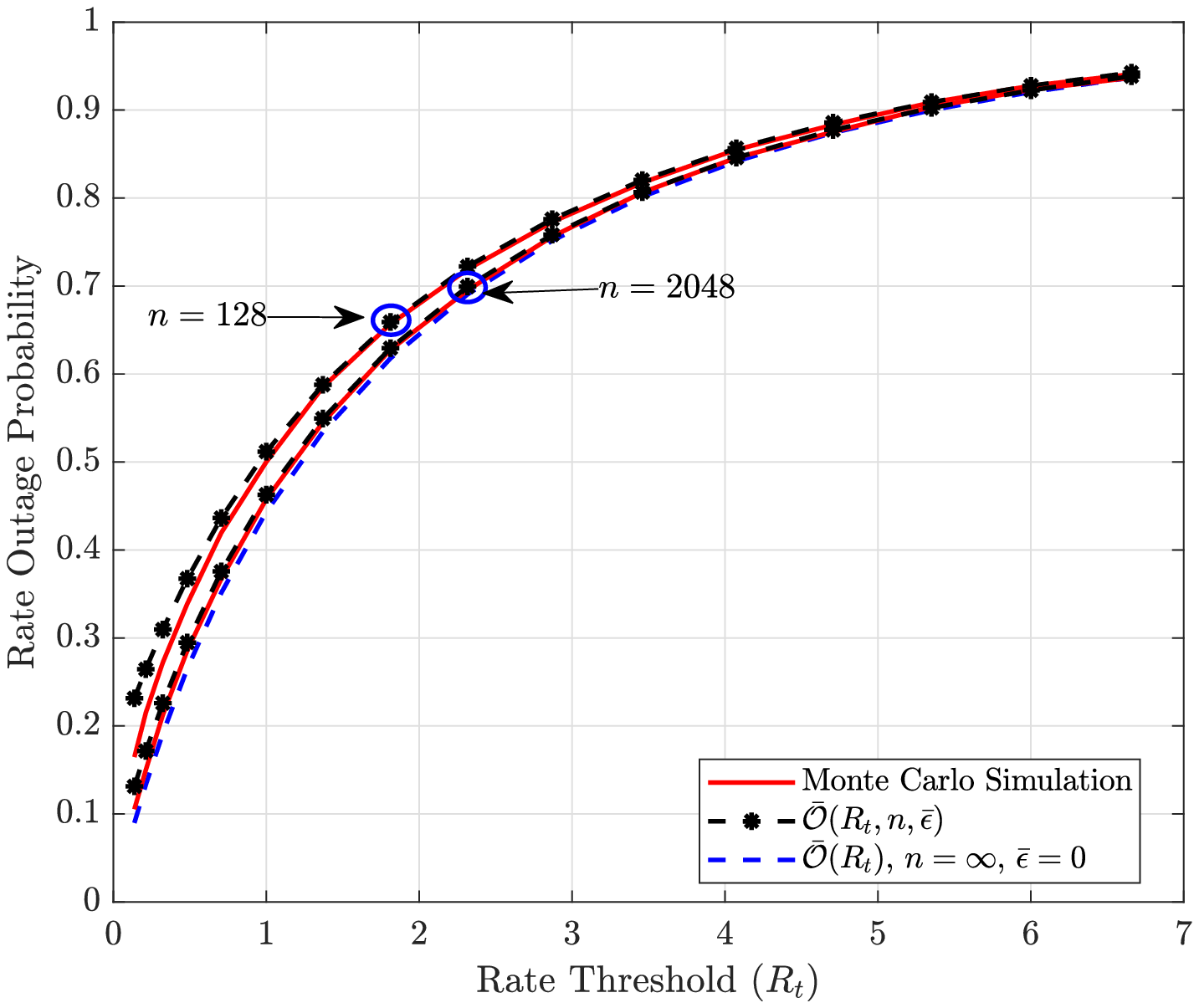}
         \caption{$\bar{\epsilon}=10^{-2}$}
         \label{fig:Outage_App_1e2}
     \end{subfigure}
     \begin{subfigure}[h]{0.5\textwidth}
         \centering
         \includegraphics[width=1\linewidth,height=6.4cm]{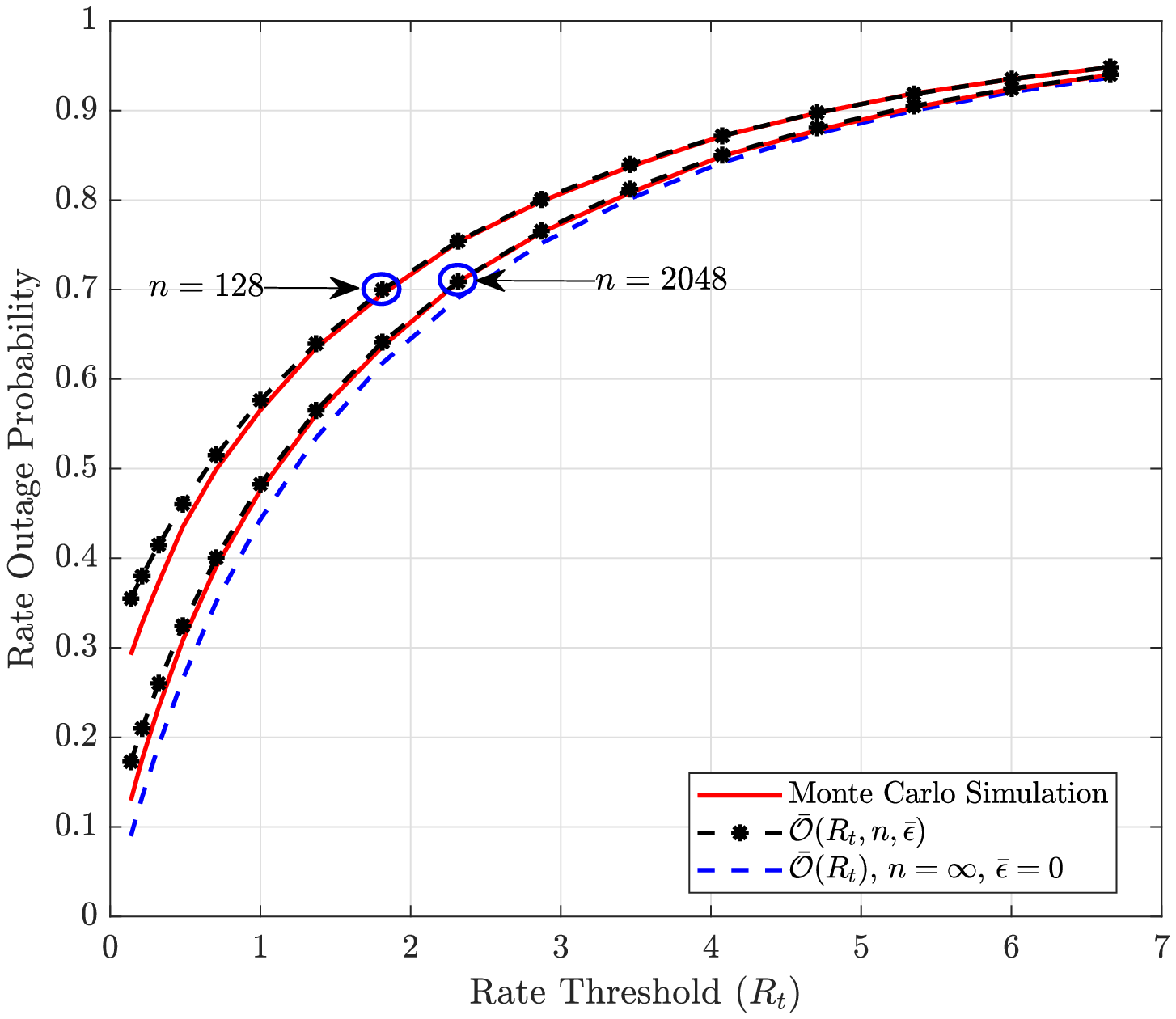}
         \caption{$\bar{\epsilon}=10^{-6}$}
         \label{fig:Outage_App_1e5}
     \end{subfigure}
     \caption{The outage probability versus the rate threshold $R_t$ at $\lambda=1\ \mathrm{BS/km^2}$, $n\in\{128,\ 2048\}$, and $\frac{\mathcal{P}}{\sigma_w^2}=0\ \text{dB}$, averaged with respect to Rayleigh distributed $r_0$.}
     \label{fig:Outage_App}
\end{figure}
\subsection{Numerical Results}\label{sec:outage_results}
Now, we provide numerical results for the bounds and approximations of the rate outage probability, reliability, and the coding rate meta-distribution. We investigate the tightness of these approximations and compare them with results from the literature for the AR.   

%practical network settings are adopted to provide insights on the accuracy of the developed {outage probability} bounds and approximations, and their comparison with results from the literature for the AR.

First, in Fig.~\ref{fig:Outage_App_r0}, the rate outage probability is plotted versus $r_0$ with $\lambda\in\{1,\ 9\}\ \mathrm{BS/km^2}$, $\eta=4$, $\bar{\epsilon}\in\{10^{-2},\ 10^{-6}\}$, $\frac{\mathcal{P}}{\sigma_w^2}=0\ \text{dB}$, and $R_t=1$ bit/transmission. The PDF of $r_0$ is plotted to illustrate the probability of occurrence of the $r_0$ values. The figure shows the Monte Carlo simulated outage probability, in addition to the upper bound in \eqref{eq:upper_bound} and the outage probability in the AR \eqref{eq:outage} which serves as a lower bound (cf. Thm. \ref{Theorem:BoundsOutage_Network_fixed_r0}). It shows that the rate outage {in the AR} underestimates the rate outage probability {in the FBR} and the gap to the actual outage probability increases as the FER threshold decreases. This gap can be as large as $30\%$ in some cases, such as with $r_0=200$ m, $\lambda=1$ BS/km$^2$, $n=128$, and $\bar{\epsilon}=10^{-6}$, where the actual outage probability is 0.13, but \eqref{eq:outage} underestimates it as 0.1.  {Fig.~\ref{fig:Outage_App_r0} also} shows that the upper bound of the outage probability expression in \eqref{eq:upper_bound} is fairly accurate to provide a convenient approximation for a broad range of $n,\lambda$, and $\bar{\epsilon}$. The outage probability is seen to increase as the FER threshold $\bar{\epsilon}$ decreases because the term $\sqrt{\frac{V(\Upsilon)}{n}}Q^{-1}(\bar{\epsilon})$ in \eqref{eq:Capacity_FB_r_0} increases as $\bar{\epsilon}$ decreases. However, at a specific $\bar{\epsilon}$, as $n$ increases, the outage probability approaches the outage probability in the AR as expected. Thus, while the results in~\cite{sawy} can be a good representative of the performance when $n$ is relatively large (Fig.~\ref{fig:Outage_App_128}), this is not the case when $n$ is small in which case our derived bounds and approximation are more accurate. Moreover, in Fig.~\ref{fig:Outage_App}, the rate outage probability approximation for a Rayleigh distributed $r_0$, provided in~\eqref{eq:outage_without_r_0}, is plotted versus the target rate $R_t$ and compared to the simulated outage probability for BS density $\lambda=1\ \text{BSs}/\text{km}^2$, blocklengths $n\in\{128,\ 2048\}$, FER threshold $\bar{\epsilon}\in\{10^{-2},\ 10^{-6}\}$, and $\frac{\mathcal{P}}{\sigma_w^2}=0\ \text{dB}$. This plot also proves the tightness of the upper bound and the inaccuracy of the AR.

 %As shown, the proposed approximation provides a tight upper bound on the outage of the large-scale network in the FBR also when $r_0$ is random. Again here we can notice the gap between the actual outage probability and the one in the AR \eqref{eq:outage} which provides an underestimate, especially when $n$ or $\bar{\epsilon}$ are small. 

%Fig~\ref{fig:Outage_vs_n} illustrates the gap between the performance of the FBR, provided in~\eqref{eq:outage_without_r_0}, and the AR versus $n$ for different values of FER $\bar{\epsilon} \in \{10^{-5},\ 10^{-8} \}$ at $R_t=0.1375$ and $\alpha=10\ \text{dB}$. It shows a gap of $> 0.25$ at small $n$, which decays as $n$ grows. The figure also shows that as the reliability constraints becomes more stringent ($\bar{\epsilon}=10^{-8}$) the gap increases. 
%\begin{figure}
 %   \centering
  %  \includegraphics[width=0.55\linewidth,height=6cm]{Outage_vs_n_new.eps}
   % \caption{The outage probability versus blocklength $n$ at $\lambda=1\ \mathrm{BS/km^2}$, $R_t=0.1375$, $\alpha=10\ \text{dB}$  }
%    \label{fig:Outage_vs_n}
%\end{figure}

In Fig.~\ref{fig:Outage_mod}, the rate outage probability {under QAM constellations} is shown for various modulation orders versus $r_0$, where the rate outage for a specific modulation scheme is defined as $\mathcal{O}^{\text{\tiny(M)}}(r_0,R_t,n,\bar{\epsilon})=\mathbb{P}(\mathcal{R}_{n,\bar{\epsilon}}^{\text{\tiny(M)}}(\Upsilon)<R_t|r_0)$. This expression is computed using Monte Carlo simulation. The figure is plotted for $n=128$, $\bar{\epsilon}=10^{-2}$, $\lambda=1$, $\frac{\mathcal{P}}{\sigma_w^2}=0\ \text{dB}$, and $R_t\in\{0.825,1.85\}$. Similar to the previous plots, the rate outage probability $\mathcal{O}(r_0,R_t,n,\bar{\epsilon})$ acts as the lower bound of the outage probability of all constellation sets. Also, the higher order modulation (i.e. $M=16$) yields an outage probability lower than the lower order modulation (i.e. $M=2$) at a fixed $R_t$ as the modulation order limits the maximum transmission rate $(\log_2(M))$. However, at lower $R_t$ as such $R_t=0.825$, most of the high order modulation schemes provide a performance close to the theoretical lower bound $\mathcal{O}(r_0,R_t,n,\bar{\epsilon})$. Fig.~\ref{fig:Outage_mod_rand_r0} shows the outage probability for random $r_0$, i.e., $\bar{\mathcal{O}}^{\text{\tiny(M)}}(R_t,n,\bar{\epsilon})=\mathbb{P}(\mathcal{R}_{n,\bar{\epsilon}}^{\text{\tiny(M)}}(\Upsilon)<R_t)$ {versus $R_t$}. As shown, the rate outage probability increases as the rate threshold $R_t$ increases, and it { jumps to 1} at a certain $R_t$ for each modulation order, {which is when $R_t$ exceeds the maximum rate achievable by the modulation order ($\log_2(M)$).}

\begin{figure}[t]
     \begin{subfigure}[h]{0.5\textwidth}
    \centering
    \includegraphics[width=1\linewidth,height=6.2cm]{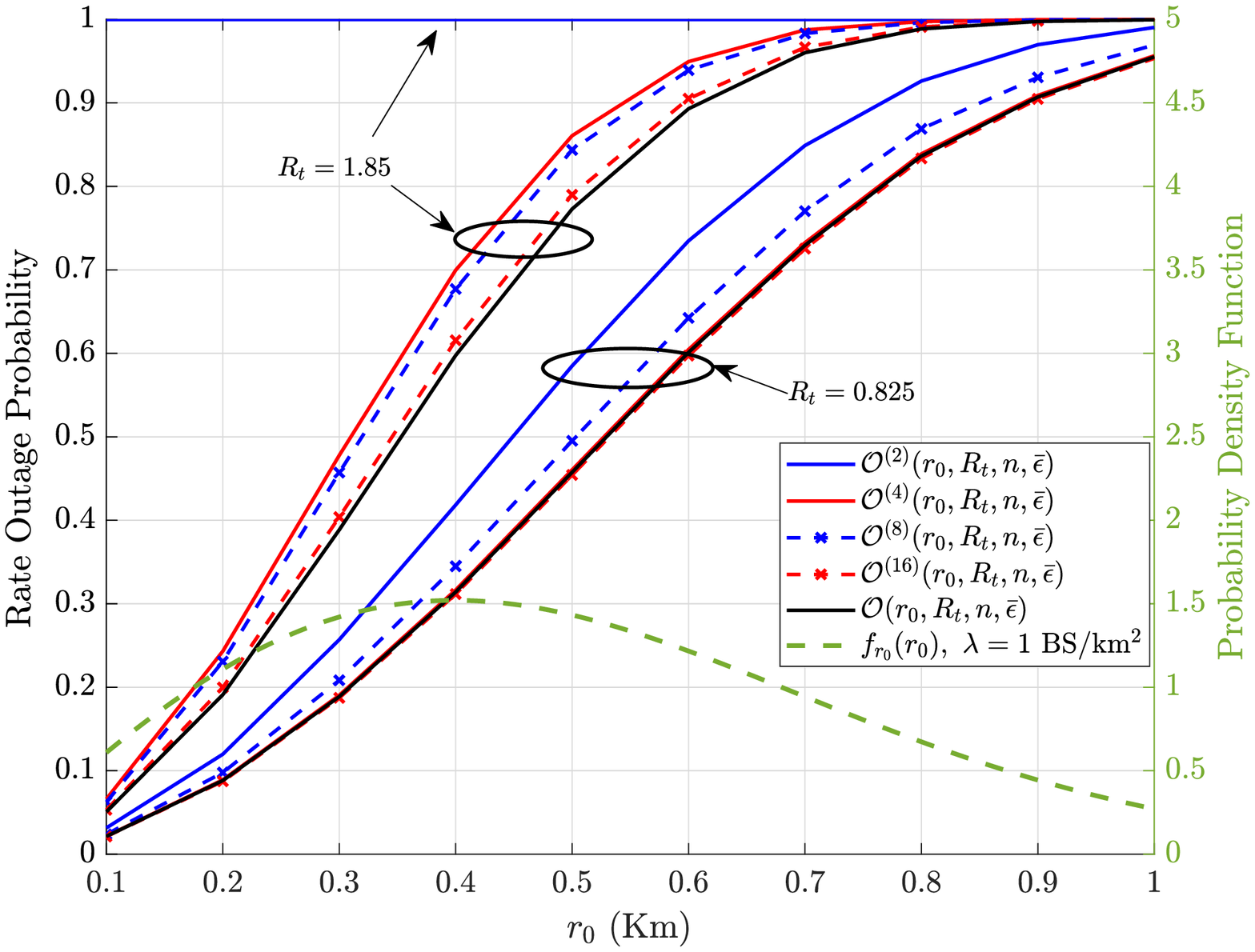}
    \caption{Outage probability versus $r_0$, for $R_t\in\{0.825,1.85\}$ }
    \label{fig:Outage_mod}
     \end{subfigure}
     \begin{subfigure}[h]{0.5\textwidth}
    \centering
    \includegraphics[width=1\linewidth,height=6.2cm]{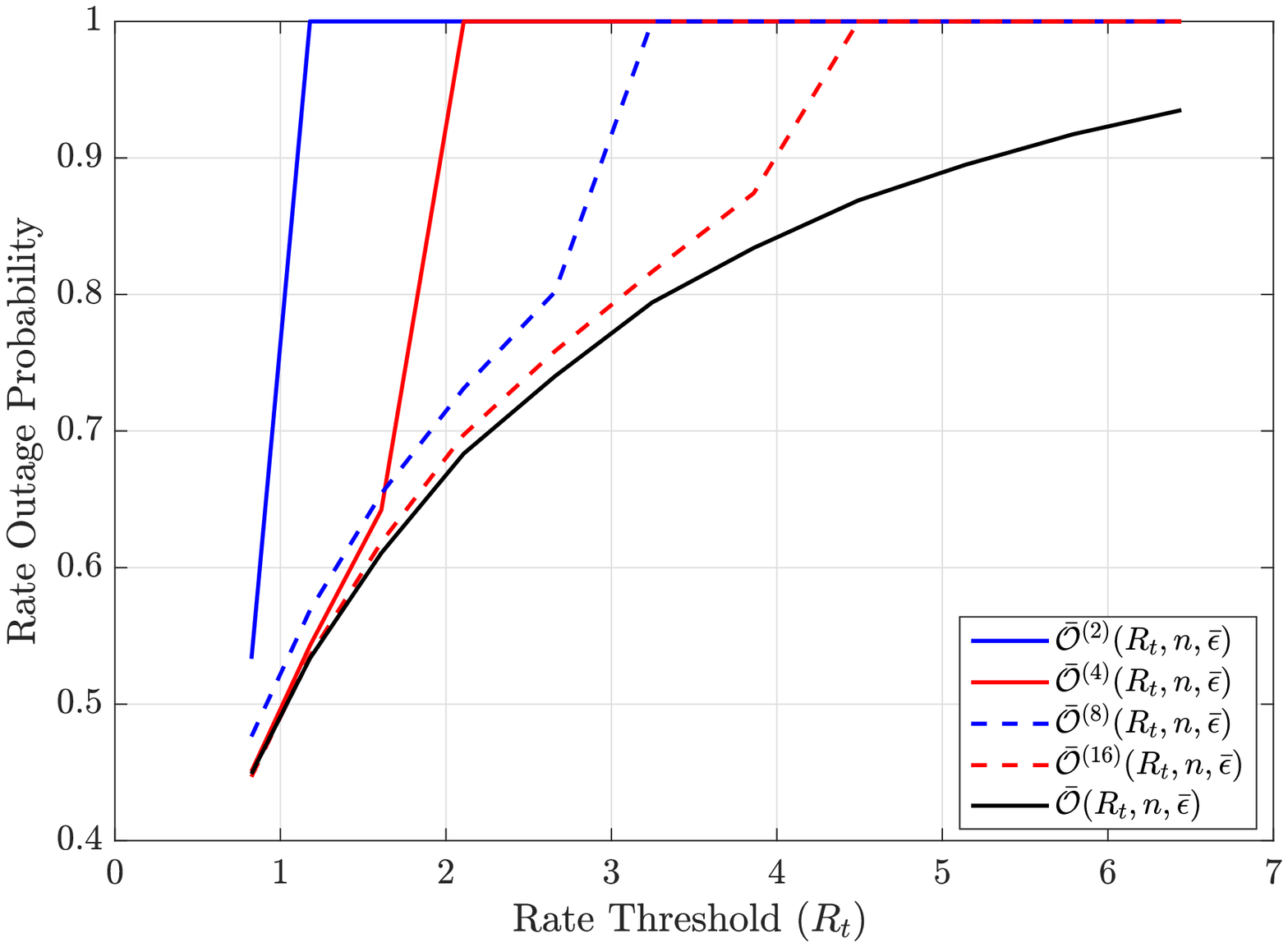}
    \caption{Outage probability versus $R_t$
    }
    \label{fig:Outage_mod_rand_r0}
     \end{subfigure}
     \caption{Outage probability using M-QAM at $\lambda=1\ \text{BS}/\text{km}^2$, $\bar{\epsilon}=10^{-2}$, $n=128$ and $\frac{\mathcal{P}}{\sigma_w^2}=0\ \text{dB}$}
     \label{fig:Outage_mod_App}
\end{figure}

The overall network reliability versus the FER threshold $\bar{\epsilon}$ is evaluated in Fig. \ref{fig:Reliability} for fixed and random $r_0$ for $n\in\{128,2048\}$ and $R_t\in\{0.1375 ,1, 3.46 \}$. It can be observed that the reliability of the network increases as the FER threshold increases until it reaches {a maximum around $\bar{\epsilon}=10^{-2}$ then} it rapidly decays. {This is because the increase of $(1- \mathcal{O}(r_0,R_t, n, \bar{\epsilon}))$ with respect to $\bar{\epsilon}$ is faster than the decrease of $(1-\bar{\epsilon})$ for small $\bar{\epsilon}$. However, for large $\bar{\epsilon}$, $(1- \mathcal{O}(r_0,R_t, n, \bar{\epsilon}))$ saturates to $1$ and $(1-\bar{\epsilon})$ dominates and contributes to the decay of the reliability at high $\bar{\epsilon}$. In other words, at $\bar{\epsilon}<10^{-2}$, the rate of increase of $(1-\mathcal{O}(r_0,R_t,n,\bar{\epsilon}))$ is higher than the rate of decay of $(1-\bar{\epsilon})$. However, at $\bar{\epsilon}>10^{-2}$, the rate of increase of $(1-\mathcal{O}(r_0,R_t,n,\bar{\epsilon}))$ is nearly zero and the rate of decay of $(1-\bar{\epsilon})$ {dominates}. Therefore, the reliability rapidly decays.} It is shown that for a fixed $R_t$, a gap exists between curves simulated under FBR and AR. This gap grows wider up to $0.15$ as FER threshold $\bar{\epsilon}$ decreases because the short blocklengths (e.g. $n=128$) are less reliable. Also, it is obvious that the reliability of the network at random $r_0$ is lower than that at fixed $r_0$ as it averages the performance of {users which are close and those that are far away from the BS}. Moreover, as $R_t$ increases the approximated reliability becomes close to the exact one. This is because the outage upper bound provided in Corollary~\ref{Theorem:BoundsOutage_Network_fixed_r0} is tight for high SINR. Hence, from Def.~\ref{def:outage}, as the $R_t$ increases, {only high SINRs can achieve the desired rate and FER threshold at the given $n$, which is where the approximation becomes tighter}. Whereas, for low $R_t$, {low SINRs can also achieve the desired performance which loosens the approximation.}

\begin{figure}[t]
     \begin{subfigure}[h]{0.52\textwidth}
    \centering
    \includegraphics[width=1\linewidth,height=6.2cm]{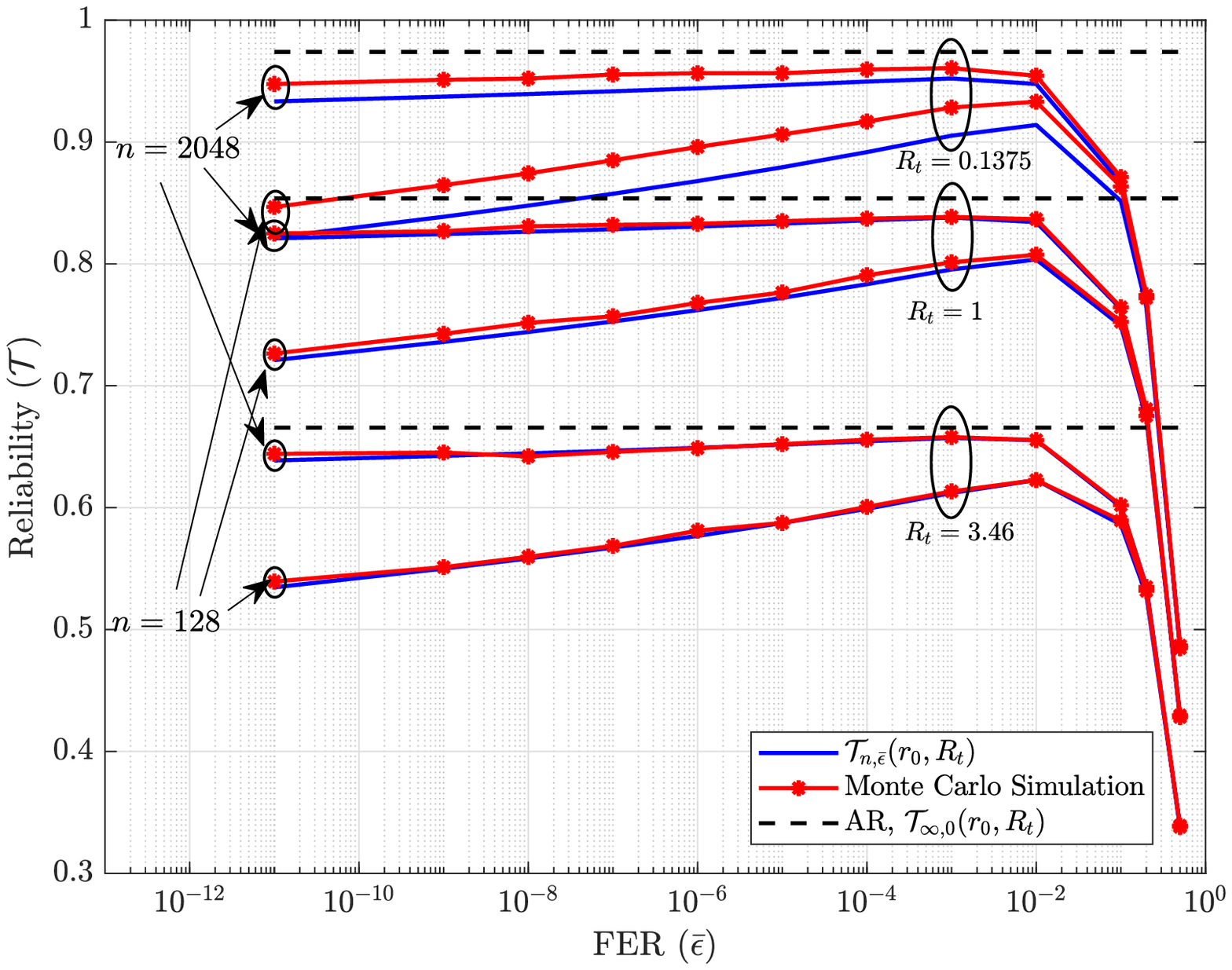}
    \caption{Fixed $r_0$: $\lambda=1$ $\text{BS}/\text{km}^2$, $\frac{\mathcal{P}}{\sigma_w^2}=0\ \text{dB}$ and $r_0=250\ m$ }
    \label{fig:Reliability_const_r0}
     \end{subfigure}
     \begin{subfigure}[h]{0.5\textwidth}
    \centering
    \includegraphics[width=1\linewidth,height=6.2cm]{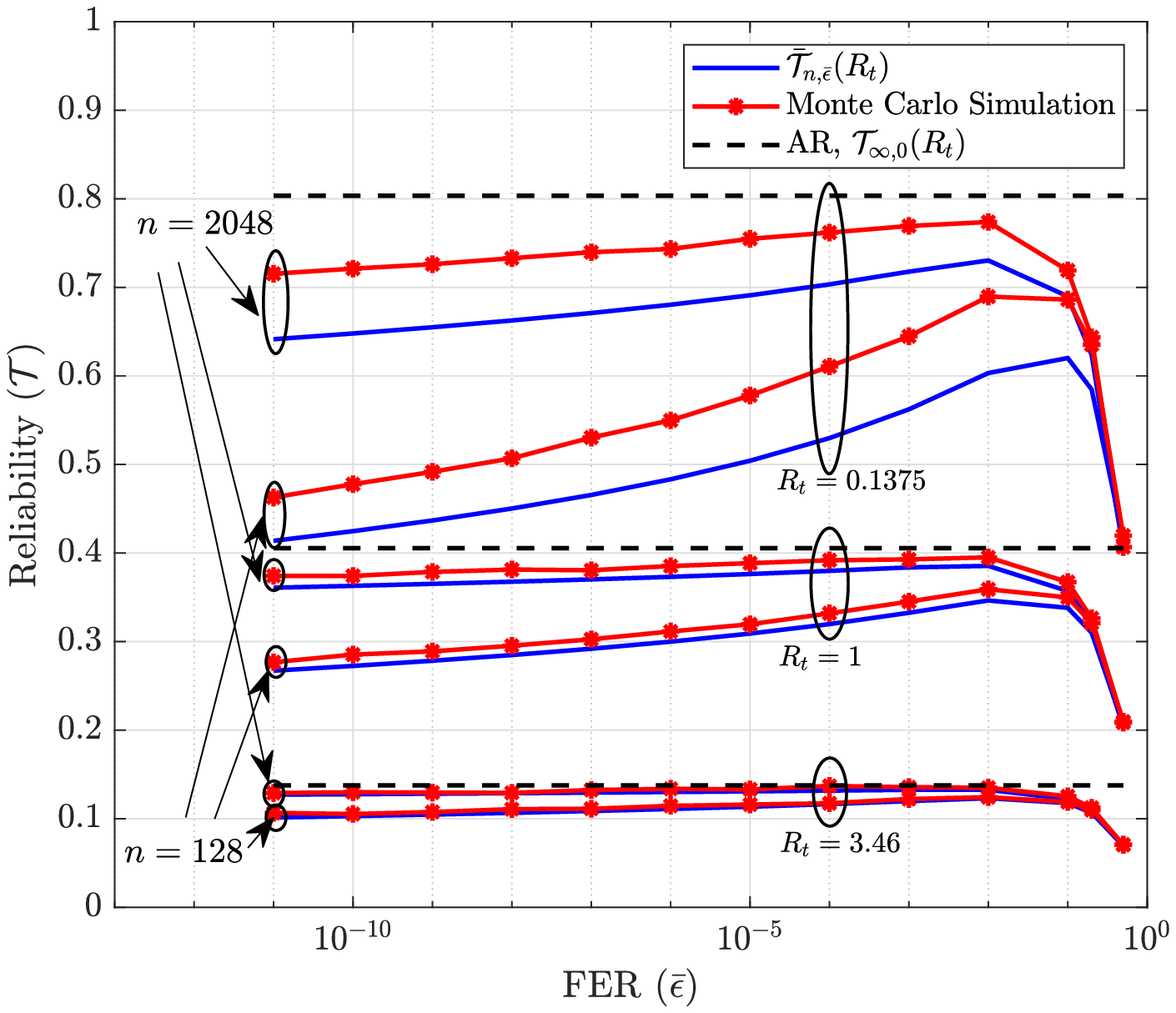}
    \caption{Random $r_0$: $\lambda=0.1$ $\text{BS}/\text{km}^2$ and $\frac{\mathcal{P}}{\sigma_w^2}=10\ \text{dB}$ }
    \label{fig:Reliability_var_r0}
     \end{subfigure}
     \caption{Reliability versus the FER threshold $\bar{\epsilon}$}
     \label{fig:Reliability}
\end{figure}

Finally, we validate the coding rate meta distribution in Fig.~\ref{fig:FBL_vs_IFB}. The coding rate meta distribution provided in~\eqref{eq:meta_exact} is simulated via Monte-Carlo simulations at $\lambda=1\ \text{BSs}/\text{km}^2$, $r_0\hspace{-0.1cm}=\hspace{-0.1cm}\ 150\ m$, an area of $500\ \text{km}^2$, a blocklength $n\hspace{-0.1cm}=\hspace{-0.1cm}128$, FER threshold $\bar{\epsilon}=10^{-5}$, and a rate threshold $R_t=\{0.1375,0.3964,1,2.0574,    3.4594\}$. It is observed that there is a significant gap between the performance in the AR and the FBR, showing that the meta distribution of the coding rate in the AR overestimates the performance of the network in the FBR. This gap can reach up to $0.09$ at $R_t=3.46$ and $p_{t}=0.9${, which means $9\%$ of the users predicted to achieve the required performance in the AR will not achieve it in the FBR.} It can also be observed that the gap increases as $R_t$ {or} $p_{t}$ increases. Thus, the AR analysis in~\cite{Meta1, Meta2, Meta3} {cannot be used to provide precise results in the FBR.} {For the sake of preciseness of the results, we have simulated the coding rate meta distribution including the noise term with noise power $N_0\hspace{-0.1cm}=\hspace{-0.1cm}-90\ \rm dBm$ \cite{Meta_noise}. The results show that the network is an interference-limited network and the noise can be neglected for the simplicity of the analysis.} In Fig.~\ref{fig:Exact_vs_App}, the approximation proposed in~\eqref{eq:P_App} is shown to {be very tight compared} to the exact coding rate meta distribution in \eqref{eq:meta_exact} for different $r_0$. The tightness of the approximation increases by decreasing $r_0$ or alternatively increasing SIR as the term $\sqrt{1-\frac{1}{(1+\Omega)^2}}$ converges to $1$ as the SIR increases. Using this approximation it is easier to find an expression for the moments of the success probability and use it to evaluate the meta distribution using the beta distribution provided in \eqref{eq:meta_beta}. Hence, the beta approximation of the meta-distribution provides a close performance to the exact and approximated meta-distribution. %This plot validates the tightness of the approximation and proves its reliability.

\begin{figure}[t]
    \centering
    \begin{subfigure}[h]{0.5\textwidth}
        \centering
        \includegraphics[width=1\textwidth,height=6.2cm]{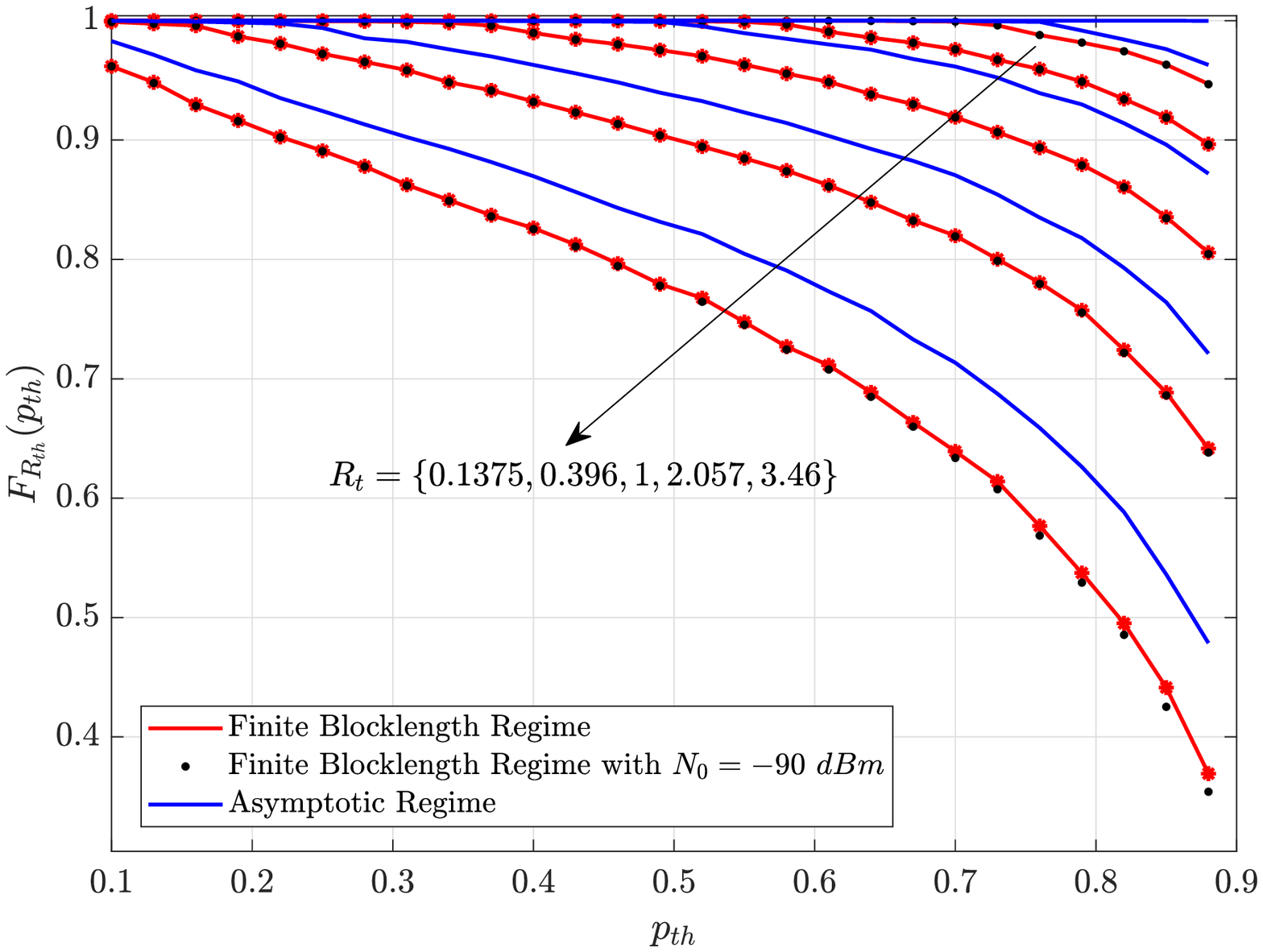}
    \caption{The FBR versus the AR at $n=128$ and $\bar{\epsilon}=10^{-5}$
    \newline}
    \label{fig:FBL_vs_IFB}
    \end{subfigure}\hfill
    \begin{subfigure}[h]{0.5\textwidth}
        \centering
        \includegraphics[width=1\textwidth,height=6.2cm]{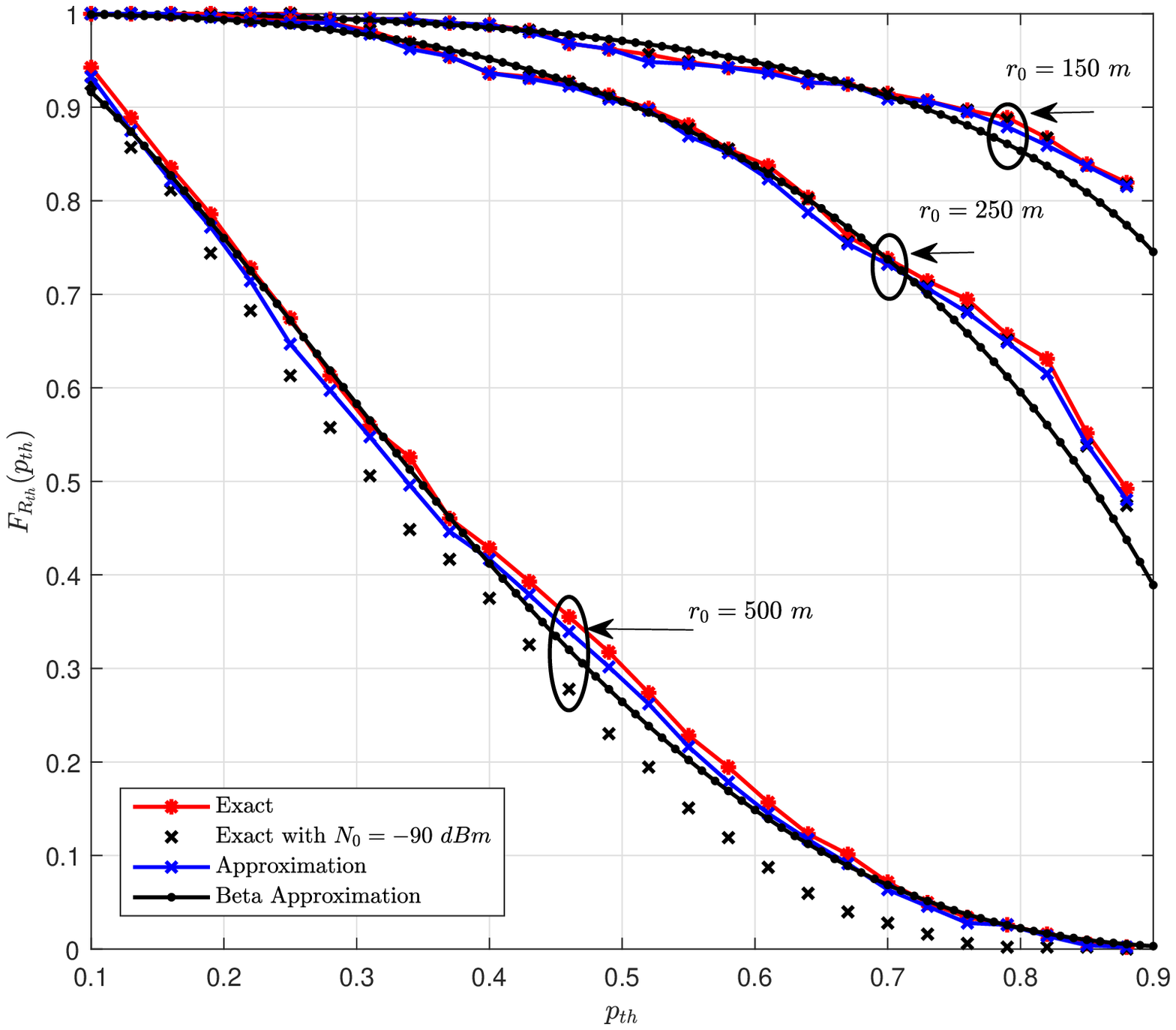}
    \caption{ Exact versus Approximation at $r_0\in\{150,\ 250,\ 500\}m$ at $R_t=1$, $n=128$, and $\bar{\epsilon}=10^{-2}$ }
    \label{fig:Exact_vs_App}
    \end{subfigure}
    \caption{The meta distribution of the {coding} rate versus $p_{t}$ }
\end{figure}

Thus, we saw that our proposed expressions provide an accurate representation of the rate outage probability and the coding rate meta distribution in the FBR, %and we studied the interaction between various parameters ($n$, $\bar{\epsilon}$, $R_t$, $\mathcal{P}$) and their effect on the performance in the FBR.
{which is more accurate than using AR expressions from the literature.}

%% file: Appendix3.tex
\section{The Average Square Root Channel Dispersion $\mathcal{V}(\alpha_0)$}\label{sec:Appendix3}

 From the definition of channel dispersion provided in \eqref{eq:Cap_polyankiy}, the average {square root} channel dispersion $\mathcal{V}(\alpha_0)$ {can be written} as follows
%\vspace{-0.2cm}
\begin{align}
  \mathcal{V}({\alpha_0})= \mathbb{E}\left\{\sqrt{ V(\Upsilon)}\right\}    &= \mathbb{E}\left\{ \log_2(e) \sqrt{\left(1-\frac{1}{(\Upsilon+1)^2}\right)}\right\}\nonumber\\
    &=\int_{0}^{\infty} (1-\mathbb{F}_{\sqrt{V}} ({v})) d{v}\label{eq:cdf_v},
\end{align}
where $\mathbb{F}_{\sqrt{V}}(v)$ is the cumulative {distribution} function (CDF) of  $\sqrt{V(\Upsilon)}$, and the last step follows as an application of Fubini's theorem~\cite{fubini}. The CDF of $\sqrt{V(\Upsilon)}$ for a given $\mathcal{B}_g$ is given by
\vspace{-0.3cm}
\begin{align}
    \mathbb{F}_{\sqrt{V}}(  {v}|\mathcal{B}_g)&= \mathbb{P}(\sqrt{V(\Upsilon)}<  {v}\big|\mathcal{B}_g)\nonumber\\
     &=\mathbb{P}\left(|h_0|^2<\left(\frac{1}{{\alpha_0}}+\frac{\mathcal{B}_g}{\mathcal{P}r_0^{-\eta}}\right){z(v)}\bigg|\mathcal{B}_g\right)\nonumber\\
     &=1-\exp\left(-\left(\frac{1}{{\alpha_0}}+\frac{\mathcal{B}_g}{\mathcal{P}r_0^{-\eta}}\right){z(v)}\right),
\end{align}
where the last step follows from the exponential distribution of $|h_0|^2$. Averaging with respect to the interference term $\mathcal{B}_g$ yields $\mathbb{E}_{\mathcal{B}_g}\{\mathbb{F}_{\sqrt{V}}(  {v}|\mathcal{B}_g)\}\hspace{-0.1cm}=\hspace{-0.1cm}\mathbb{E}_{\mathcal{B}_g}\left\{1-\exp\hspace{-0.1cm}\left(\frac{-{z(v)}}{{\alpha_0}}\right) \exp\left(-\frac{\mathcal{B}_g r_0^{\eta}}{\mathcal{P} }{z(v)}\right)\right\}$. Hence, the CDF of the square root of the channel dispersion is given as
\begin{align}
\label{eq:avg_cdf_v}
\mathbb{F}_{\sqrt{V}}(  {v})=
\hspace{-0.1cm}1-\exp\hspace{-0.1cm}\left(\frac{-{z(v)}}{{\alpha_0}}\right) \mathcal{L}_{\mathcal{B}_g}\left\{\frac{r_0^{\eta}}{\mathcal{P} }{z(v)}\right\}, \ \ \ \ \ \ \ \   0\leq v< \log_2(e),
\end{align}
which is obtained using the definition of Laplace transform $\left(\mathcal{L}_x\{u\}=\mathbb{E}\left\{e^{-ux} \right\}\right)$.
By substituting \eqref{eq:avg_cdf_v} in \eqref{eq:cdf_v}, we obtain
%\vspace{-0.3cm}
\begin{align}
   \mathcal{V}(\alpha_0)=\int_{0}^{\log_2(e)} \hspace{-0.5cm}\exp\left(\frac{-{z(v)}}{{\alpha_0}}\right)\mathcal{L}_{\mathcal{B}_g}\left\{\frac{r_0^{\eta}}{\mathcal{P}}{z(v)}\right\} d  {v},\nonumber
\end{align}
as defined in Theorem 1.

%% file: AppendixA_new.tex
\section{The Channel Dispersion for A Large-Scale Network {Under a Finite Constellation}}
\label{appendix}
{The channel dispersion is defined as the conditional variance of the information density conditioned on the input distribution} as follows
\begin{align}\label{eq:1}
   \hspace{-0.2cm} V_{\text{\tiny M}}(\Upsilon)&=\mathbb{E}_S\left\{\mathbb{V}ar\left\{ \log_2  \frac{P_{Y|S}(y|s)}{P_Y(y)}\mid S\right\}\right\}\nonumber\\
   &=\frac{1}{M}\sum_{m=1}^{M} \mathbb{E}\left\{ \log_2^2  \frac{P_{Y|S}(y|\underline{s}_m)}{P_Y(y)}\right\} -\frac{1}{M}\sum_{m=1}^{M}\mathbb{E}\left\{ \log_2  \frac{P_{Y|S}(y|\underline{s}_m)}{P_Y(y)}\right\}^2. \nonumber
\end{align}
{where $\underline{s}_m=[\mathcal{R}e\{s_m\},\mathcal{I}m\{s_m\}]^T$ is the 2-D real-valued vector form of the complex-valued symbol $s_m$ of an $M$-ary constellation.}

Define $I_1=\mathbb{E}\left\{ \log_2^2  \frac{P_{Y|S}(y|\underline{s}_m)}{P_Y(y)}\right\}$ and $I_2=\mathbb{E}\left\{ \log_2  \frac{P_{Y|S}(y|\underline{s}_m)}{P_Y(y)}\right\}^2$. The first term $I_1$ is derived as follows
\begin{align}
    I_1 =\mathbb{E}\left\{ \log_2^2  P_{Y|S}(y|\underline{s}_m)\right\}-2 \mathbb{E}\left\{\log_2 P_{Y|S}(y|\underline{s}_m)\log_2 P_Y(y) \right\}+\mathbb{E}\left\{\log_2^2P_Y(y) \right\}=I_{11}-2 I_{12}+I_{13}.\nonumber
\end{align}

The term $I_{11}$ is given by
\begin{align}
\hspace{-4.8cm}  I_{11}=  \int_{{\mathbb{R}^2}} \frac{e^{ -\frac{\|{\underline{y}}-h_0 \sqrt{\mathcal{P}} r_0^{-\eta/2} {\underline{s}_m}\|^2}{N_0+\mathcal{B}} }}{{\pi (N_0+\mathcal{B})}} \log_2^2\frac{ e^{ -\frac{\|{\underline{y}}-h_0 \sqrt{\mathcal{P}} r_0^{-\eta/2} {\underline{s}_m}\|^2}{N_0+\mathcal{B}}}}{{\pi (N_0+\mathcal{B})}} \ d{\underline{y}},\nonumber
\end{align}
{where $\underline{y}=[y_1,y_2]^T$.} Define ${\underline{t}}=\frac{{\underline{y}}-h_0 \sqrt{\mathcal{P}} r_0^{-\eta/2} {\underline{s}_m}}{\sqrt{N_0+\mathcal{B}}}$, hence by substitution, $I_{11}$ is given by
%\vspace{-0.5cm}
\begin{align}
     I_{11}  &= \int_{{\mathbb{R}^2}}\frac{e^{-\|{\underline{t}}\|^2}}{\pi} %\hspace{-0.07cm}
 \left( \log_2^2\frac{1}{\pi(N_0+\mathcal{B})} %\hspace{-0.05cm}
 +2\|{\underline{t}}\|^2\log_2\left(e\right)\log_2 \frac{1}{\pi(N_0+\mathcal{B})}  +\log_2^2e^{\|{\underline{t}}\|^2}  \right) d{\underline{t}}\nonumber\\
&=\log_2^2\frac{1}{\pi(N_0+\mathcal{B})}-2 \log_2(e) \log_2\frac{1}{\pi(N_0+\mathcal{B})} +3\log_2^2(e)\nonumber.
\end{align}
The term $I_{12}$ is given by
   \begin{align*}
   I_{12}= \int_{{\mathbb{R}^2}} \frac{e^{-\frac{\|{\underline{y}}-h_0\sqrt{\mathcal{P}}r_0^{\eta/2}{\underline{s}_m}\|^2}{(N_0+\mathcal{B})} }}{\pi(N_0+\mathcal{B})} \log_2\left( \frac{1}{M} \sum_{l=1}^{M}  \frac{e^{-\frac{\|{\underline{y}}-h_0\sqrt{\mathcal{P}}r_0^{\eta/2}{\underline{s}_l}\|^2}{(N_0+\mathcal{B)}} }}{\pi(N_0+\mathcal{B})} \right) \log_2\left( \frac{e^{\frac{-\|{\underline{y}}-h_0\sqrt{\mathcal{P}}r_0^{\eta/2}{\underline{s}_m}\|^2}{(N_0+\mathcal{B)}} }}{\pi(N_0+\mathcal{B})}
   \right) d{\underline{y}},
   \end{align*}
{and is divided into} four terms {as $I_{12}=I_{12}^{(1)}+I_{12}^{(2)}+I_{12}^{(3)}+I_{12}^{(4)}$,} where
\begin{align}
 I_{12}^{(1)}&= \int_{{\mathbb{R}^2}} \frac{e^{-\frac{\|{\underline{y}}-h_0\sqrt{\mathcal{P}}r_0^{\eta/2}{\underline{s}_m}\|^2}{(N_0+\mathcal{B})} }}{\pi(N_0+\mathcal{B})}  \log_2 \frac{1}{\pi M(N_0+\mathcal{B})} \log_2 \frac{1}{\pi(N_0+\mathcal{B})} d{\underline{y}}\hspace{6cm}\nonumber\\
  &=\log_2 \frac{1}{\pi M (N_0+\mathcal{B})} \log_2 \frac{1}{\pi (N_0+\mathcal{B})},\nonumber
    \end{align}
    
   \begin{align}
   I_{12}^{(2)}&= \int_{{\mathbb{R}^2}} \frac{e^{-\frac{\|{\underline{y}}-h_0\sqrt{\mathcal{P}}r_0^{\eta/2}{\underline{s}_m}\|^2}{(N_0+\mathcal{B})}}}{\pi(N_0+\mathcal{B})}  \log_2\frac{1}{\pi M(N_0+\mathcal{B})} 
      \log_2  e^{\frac{-\|{\underline{y}}-h_0\sqrt{\mathcal{P}}r_0^{\eta/2}{\underline{s}_m}\|^2}{(N_0+\mathcal{B})} }d{\underline{y}}\hspace{2cm}\ \ \ \ \ \ \ \ \ \ \ \ \ \ \ \ 
 \ \ \nonumber\\
      &= -\log_2(e) \log_2 \frac{1}{\pi M (N_0+\mathcal{B})},\nonumber
    \end{align}
    \begin{align}
        I_{12}^{(3)}&= \int_{{\mathbb{R}^2}} \frac{e^{-\frac{\|{\underline{y}}-h_0\sqrt{\mathcal{P}}r_0^{\eta/2}{\underline{s}_m}\|^2}{(N_0+\mathcal{B})} }}{\pi(N_0+\mathcal{B})}  \log_2\frac{1}{\pi (N_0+\mathcal{B})}  \log_2 \sum_{l=1}^{M}  e^{-\frac{\|{\underline{y}}-h_0\sqrt{\mathcal{P}}r_0^{\eta/2}{\underline{s}_l}\|^2}{(N_0+\mathcal{B})} }\ d{\underline{y}}\hspace{6cm}\nonumber\\
        &\stackrel{\text{(i)}}{=}\frac{1}{\pi} \int_{{\mathbb{R}^2}} e^{-\|\underline{t}\|^2} \log_2\frac{1}{\pi(N_0+\mathcal{B})} \log_2\sum_{l=1}^{M} e^{-\|\underline{t}\|^2-2 \underline{t}^T(\underline{s}_m-\underline{s}_l)\frac{h_0 \sqrt{\mathcal{P}} r_0^{-\eta/2}}{\sqrt{N_0+\mathcal{B}}}-\frac{|h_0|^2 {\mathcal{P}} r_0^{-\eta}}{{N_0+\mathcal{B}}} \left\|\underline{s}_m-\underline{s}_l\right\|^2  }    d\underline{t}\nonumber\\
          &\stackrel{\text{(ii)}}{=}\frac{1}{\pi} \int_{{\mathbb{R}^2}} \hspace{-0.3cm}e^{-\|\underline{t}\|^2} \log_2 \frac{1}{\pi(N_0+\mathcal{B})} \log_2 \sum_{l=1}^{M} e^{-\left(\|\underline{t}\|^2+2\sqrt{\Upsilon} \underline{t}^T\underline{d}_{ml}+\Upsilon\left\| \underline{d}_{ml}\right\|^2 \right) }   d\underline{t}\nonumber\\
    & =-\log_2(e)\log_2\frac{1}{\pi (N_0+\mathcal{B})} \hspace{-0.1cm} +\mathbb{E}\left\{ \log_2 \frac{1}{\pi (N_0+\mathcal{B})} \log_2\sum_{l=1}^{M} e^{-2\sqrt{\Upsilon} \underline{t}^T \underline{d}_{ml}-\Upsilon\left\| \underline{d}_{ml}\right\|^2  }   \right\}\nonumber,
    \end{align}
 where $(\rm i)$ follows from a change of variable ${\underline{t}}=\frac{{\underline{y}}-h_0 \sqrt{\mathcal{P}} r_0^{-\eta/2} {\underline{s}_m}}{\sqrt{N_0+\mathcal{B}}}$, and $(\rm ii)$ follows by using $ \underline{ d}_{ml}=\underline{s}_m-\underline{s}_l$ and $\Upsilon=\frac{|h_0|^2 \mathcal{P} r_0^{-\eta}}{{N_0+\mathcal{B}}}$, and $I_{12}^{(4)}$ is given by 
 
      \begin{align}
        I_{12}^{(4)}&= \int_{{\mathbb{R}^2}} \frac{e^{-\frac{\|{\underline{y}}-h_0\sqrt{\mathcal{P}}r_0^{\eta/2}{\underline{s}_m}\|^2}{(N_0+\mathcal{B})}}}{\pi(N_0+\mathcal{B})}   \log_2  e^{\frac{-\|{\underline{y}}-h_0\sqrt{\mathcal{P}}r_0^{\eta/2}{\underline{s}_m}\|^2}{(N_0+\mathcal{B})} }  \log_2 \sum_{l=1}^{M}  e^{-\frac{\|{\underline{y}}-h_0\sqrt{\mathcal{P}}r_0^{\eta/2}{\underline{s}_l}\|^2}{(N_0+\mathcal{B})} } d{\underline{y}}\hspace{3cm}\nonumber\\
        & \stackrel{\text{(i)}}{=} \int_{{\mathbb{R}^2}} \frac{e^{-\|\underline{t}\|^2 }}{\pi}   \log_2  e^{-\|\underline{t}\|^2}
      \log_2 \sum_{l=1}^{M}   e^{-\left(\|\underline{t}\|^2+2 \underline{t}^T\underline{d}_{ml}\sqrt{\Upsilon}+\Upsilon \left\|\underline{d}_{ml}\right\|^2 \right) }  d{\underline{t}}\nonumber\\
        &=3 \log_2^2(e)-\log_2(e) \mathbb{E}_{\underline{t}}\left\{ \|\underline{t}\|^2 \log_2 \sum_{l=1}^{M} e^{-2 \sqrt{\Upsilon} \underline{t}^T \underline{d}_{ml} - \Upsilon \|\underline{d}_{ml}\|^2  }  \right\},\nonumber
    \end{align}
where $(\rm i)$ is obtained by the
substitution of $\underline{t}$, $\underline{d}_{ml}$, and $\Upsilon$. 

The term $I_{13}$ is given by
\begin{align*}
 \hspace{-2.5cm}   I_{13}=\int_{{\mathbb{R}^2}} \frac{e^{-\frac{\|{\underline{y}}-h_0\sqrt{\mathcal{P}}r_0^{\eta/2}{\underline{s}_m}\|^2}{(N_0+\mathcal{B})} }}{\pi(N_0+\mathcal{B})}  \hspace{-0.1cm}\left(\log_2 \frac{1}{M\pi(N_0+\mathcal{B})} + \hspace{-0.07cm}\log_2\sum_{l=1}^{M}e^{\frac{-\|{\underline{y}}-h_0\sqrt{\mathcal{P}}r_0^{\eta/2}{\underline{s}_l}\|^2}{(N_0+\mathcal{B})} }\right)^2
   d{\underline{y}},
\end{align*}
{and is divided into } three terms {as $I_{13}=I_{13}^{(1)}+I_{13}^{(2)}+I_{13}^{(3)}$}, where
\begin{align}
    I_{13}^{(1)}&= \int_{{\mathbb{R}^2}} \frac{e^{-\frac{\|{\underline{y}}-h_0\sqrt{\mathcal{P}}r_0^{\eta/2}{\underline{s}_m}\|^2}{(N_0+\mathcal{B})} }}{\pi (N_0+\mathcal{B})}  \log_2^2 \frac{1}{\pi M (N_0+\mathcal{B})} \ d{\underline{y}}\hspace{3cm} \ \ \ \ \ \ \ \ \ \ \  \ \ \ \ \ \ \ \ \ \ \ \ \ \ \ \ \ \ \ \ \ \ \ \ \ \ \nonumber\\
     &=\log_2^2 \frac{1}{\pi M (N_0+\mathcal{B})},\nonumber
    \end{align}
    \begin{align}
    I_{13}^{(2)}&={2} \int_{{\mathbb{R}^2}} \frac{ e^{-\frac{\|{\underline{y}}-h_0\sqrt{\mathcal{P}}r_0^{\eta/2}{\underline{s}_m}\|^2}{(N_0+\mathcal{B})} }}{\pi(N_0+\mathcal{B})} \log_2 \frac{1}{\pi M(N_0+\mathcal{B})}\log_2\sum_{l=1}^{M}e^{\frac{-\|{\underline{y}}-h_0\sqrt{\mathcal{P}}r_0^{\eta/2}{\underline{s}_l}\|^2}{(N_0+\mathcal{B})} }\  d{\underline{y}}\nonumber\\
     &=-2 \log_2(e) \log_2\frac{1}{\pi M(N_0+\mathcal{B})}  + 2 \log_2\frac{1}{\pi M(N_0+\mathcal{B})}  \mathbb{E}\left\{  \log_2 \sum_{l=1}^{M} e^{-2 \sqrt{\Upsilon} {\underline{t}^T}  \underline{ d}_{ml}- \Upsilon \| \underline{ d}_{ml}\|^2  }  \right\},\nonumber
         \end{align}
\begin{align}
    I_{13}^{(3)}&=\int_{{\mathbb{R}^2}} \frac{1}{\pi(N_0+\mathcal{B})} e^{-\frac{\|{\underline{y}}-h_0\sqrt{\mathcal{P}}r_0^{\eta/2}{\underline{s}_m}\|^2}{(N_0+\mathcal{B})} }
    \log_2^2 \sum_{l=1}^{M}e^{\frac{-\|{\underline{y}}-h_0\sqrt{\mathcal{P}}r_0^{\eta/2}{\underline{s}_l}\|^2}{(N_0+\mathcal{B})} }\ d{\underline{y}}\nonumber\\
   & = 3\log_2^2(e)- 2 \log_2(e)\ \mathbb{E}\left\{ \|{\underline{t}}\|^2 \log_2\sum_{l=1}^{M} e^{-2 \sqrt{\Upsilon} {\underline{t}^T} \underline{ d}_{ml}- \Upsilon \|\underline{ d}_{ml}\|^2  }  \right\}+ \mathbb{E}\left\{  \log_2\sum_{l=1}^{M} e^{-2 \sqrt{\Upsilon} {\underline{t}^T} \underline{ d}_{ml}- \Upsilon \|\underline{ d}_{ml}\|^2  }  \right\},\nonumber
   \end{align}

By adding $I_{11}$, $I_{12}$, and $I_{13}$, $I_1$ is given by
\begin{align}
    I_1=\log_2^2\frac{1}{M}\hspace{-0.1cm}+2 \log_2\frac{1}{M}  \mathbb{E}\left\{  \log_2 \sum_{l=1}^{M} e^{-2 \sqrt{\Upsilon} {\underline{t}^T} {\underline{d}_{ml}}- \Upsilon \|{\underline{d}_{ml}}\|^2  } \right\} + \mathbb{E}\left\{  \log_2^2 \sum_{l=1}^{M} e^{-2 \sqrt{\Upsilon} {\underline{t}^T}{\underline{d}_{ml}}- \Upsilon \|{\underline{d}_{ml}}\|^2 }  \right\}, \nonumber
\end{align}
The second term $I_2$ is the square of the relative entropy given as 
\begin{align}
I_2=\log_2^2 (M)-2 \log_2( M )\mathbb{E}\left\{  \log_2 \sum_{l=1}^{M} e^{-2 \sqrt{\Upsilon} {\underline{t}^T} {\underline{d}_{ml}}- \Upsilon \|{\underline{d}_{ml}}\|^2 }  \right\}
+\mathbb{E}^2\left\{  \log_2 \sum_{l=1}^{M} e^{-2 \sqrt{\Upsilon} {\underline{t}^T} {\underline{d}_{ml}}- \Upsilon \|{\underline{d}_{ml}}\|^2  }  \right\}. \nonumber
\end{align}
To compute the {average square-root} channel dispersion, $I_2$ is subtracted from $I_1$ and then averaged over all constellation symbols, leading to the following expression as indicated in \eqref{eq:V_mod}
\begin{align}
    V_{\text{\tiny M}}(\Upsilon)=\frac{1}{M}\sum_{m=1}^{M}\left\{\mathbb{E}_{\underline{t}}\left\{  \log_2^2 \sum_{l=1}^{M} e^{-2 \sqrt{\Upsilon} {\underline{t}^T} {\underline{d}_{ml}}- \Upsilon \|{\underline{d}_{ml}}\|^2  }  \right\} -\mathbb{E}_{\underline{t}}\left\{  \log_2 \sum_{l=1}^{M} e^{-2 \sqrt{\Upsilon} {\underline{t}^T} {\underline{d}_{ml}}- \Upsilon \|{\underline{d}_{ml}}\|^2  }  \right\}^2\right\}. \nonumber
\end{align}
\vspace{-1cm}

%% file: Appendix2.tex
\vspace{-0.5cm}
\section{The SINR distribution at a Fixed $r_0$}
\label{appendix2}
The distribution of the SINR $\Upsilon$ at a fixed $r_0$ is derived by defining the CDF as follows
\begin{align}
   \mathbb{F}_{\Upsilon}(\Upsilon|r_0)&=\mathbb{P}\left(|h_0|^2<\frac{\Upsilon(\mathcal{B}+\sigma_w^2)}{\mathcal{P}r_0^{-\eta}}\right)\nonumber\\
    &=1-\mathbb{E}\left\{e^{\frac{-\Upsilon\mathcal{B}}{\mathcal{P}r_0^{-\eta}}}\right\} e^{\frac{-\Upsilon \sigma_w^2}{\mathcal{P}r_0^{-\eta}}}\nonumber\\
    &= 1- \mathcal{L}_{{\mathcal{B}}}\left\{\frac{\Upsilon}{\mathcal{P} r_0^{-\eta}} \right\}  e^{\frac{-\Upsilon \sigma_w^2}{\mathcal{P}r_0^{-\eta}}}\nonumber,
\end{align}
using the Gamma approximation provided in \eqref{eq:Gamma_app}. The PDF of the SINR $\Upsilon$ at a fixed $r_0$ is given by differentiating the CDF as follows
\begin{align}
    f(\Upsilon|r_0)=\frac{ e^{\frac{-\Upsilon \sigma_w^2}{\mathcal{P}r_0^{-\eta}}}}{\mathcal{P}r_0^{-\eta}} \frac{(1+ \frac{\theta r_0^{\eta} \Upsilon}{\mathcal{P}})\sigma_w^2+q\theta}{(1+ \frac{\theta r_0^{\eta} \Upsilon}{\mathcal{P}})^{q+1}},
\end{align}
this proves \eqref{eq:pdf_SINR}.

%% file: IEEE_TCOM.bbl
\begin{thebibliography}{10}
\providecommand{\url}[1]{#1}
\csname url@samestyle\endcsname
\providecommand{\newblock}{\relax}
\providecommand{\bibinfo}[2]{#2}
\providecommand{\BIBentrySTDinterwordspacing}{\spaceskip=0pt\relax}
\providecommand{\BIBentryALTinterwordstretchfactor}{4}
\providecommand{\BIBentryALTinterwordspacing}{\spaceskip=\fontdimen2\font plus
\BIBentryALTinterwordstretchfactor\fontdimen3\font minus
  \fontdimen4\font\relax}
\providecommand{\BIBforeignlanguage}[2]{{%
\expandafter\ifx\csname l@#1\endcsname\relax
\typeout{** WARNING: IEEEtran.bst: No hyphenation pattern has been}%
\typeout{** loaded for the language `#1'. Using the pattern for}%
\typeout{** the default language instead.}%
\else
\language=\csname l@#1\endcsname
\fi
#2}}
\providecommand{\BIBdecl}{\relax}
\BIBdecl

\bibitem{mypaper}
N.~Hesham and A.~Chaaban, ``On the performance of large-scale wireless networks
  in the finite block-length regime,'' in \emph{ICC 2021 - IEEE Int. Conf.
  Commun.}, 2021, pp. 1--6.

\bibitem{Short_Motivations}
Z.~Ma, M.~Xiao, Y.~Xiao, Z.~Pang, H.~V. Poor, and B.~Vucetic,
  ``High-reliability and low-latency wireless communication for internet of
  things: Challenges, fundamentals, and enabling technologies,'' \emph{IEEE
  Internet of Things Journal}, vol.~6, no.~5, pp. 7946--7970, 2019.

\bibitem{5G_2}
H.~Yu, H.~Lee, and H.~Jeon, ``What is {5G}? emerging {5G} mobile services and
  network requirements,'' \emph{Sustainability}, vol.~9, no.~10, p. 1848, 2017.

\bibitem{IoT}
A.~J. Pinheiro, J.~de~M.~Bezerra, C.~A. Burgardt, and D.~R. Campelo,
  ``Identifying {IoT} devices and events based on packet length from encrypted
  traffic,'' \emph{Comput. Commun.}, vol. 144, pp. 8 -- 17, 2019.

\bibitem{5G_1}
\BIBentryALTinterwordspacing
W.~Khalid, H.~Yu, R.~Ali, and R.~Ullah, ``Advanced physical-layer technologies
  for beyond {5G} wireless communication networks,'' \emph{Sensors}, vol.~21,
  no.~9, 2021. [Online]. Available:
  \url{https://www.mdpi.com/1424-8220/21/9/3197}
\BIBentrySTDinterwordspacing

\bibitem{polyanski}
Y.~Polyanskiy, H.~V. Poor, and S.~Verd{\'u}, ``Channel coding rate in the
  finite blocklength regime,'' \emph{IEEE Trans. Inf. Theory}, vol.~56, no.~5,
  pp. 2307--2359, 2010.

\bibitem{BlockFadingPolyankiy}
W.~{Yang}, G.~{Durisi}, T.~{Koch}, and Y.~{Polyanskiy}, ``Block-fading channels
  at finite blocklength,'' in \emph{Int. Symp. Wireless Commun. Sys.}, 2013,
  pp. 1--4.

\bibitem{paper1_SG}
H.~Q. {Nguyen}, F.~{Baccelli}, and D.~{Kofman}, ``A stochastic geometry
  analysis of dense {IEEE} 802.11 networks,'' in \emph{IEEE Int. Conf. Comput.
  Commun.}, 2007, pp. 1199--1207.

\bibitem{paper2_SG}
M.~{Haenggi}, J.~G. {Andrews}, F.~{Baccelli}, O.~{Dousse}, and
  M.~{Franceschetti}, ``Stochastic geometry and random graphs for the analysis
  and design of wireless networks,'' \emph{IEEE J. Sel. Areas Commun.},
  vol.~27, no.~7, pp. 1029--1046, 2009.

\bibitem{paper3_SG}
C.-H. Lee, C.-Y. Shih, and Y.-S. Chen, ``Stochastic geometry based models for
  modeling cellular networks in urban areas,'' \emph{Wireless Netw.}, vol.~19,
  no.~6, pp. 1063--1072, 2013.

\bibitem{SG}
Y.~Hmamouche, M.~Benjillali, S.~Saoudi, H.~Yanikomeroglu, and M.~D. Renzo,
  ``New trends in stochastic geometry for wireless networks: A tutorial and
  survey,'' \emph{Proc. IEEE}, vol. 109, no.~7, pp. 1200--1252, 2021.

\bibitem{sawy}
H.~ElSawy, A.~Sultan-Salem, M.-S. Alouini, and M.~Z. Win, ``Modeling and
  analysis of cellular networks using stochastic geometry: A tutorial,''
  \emph{IEEE Commun. Surveys Tuts.}, vol.~19, no.~1, pp. 167--203, 2016.

\bibitem{sawy2}
H.~{ElSawy}, E.~{Hossain}, and M.~{Haenggi}, ``Stochastic geometry for
  modeling, analysis, and design of multi-tier and cognitive cellular wireless
  networks: A survey,'' \emph{IEEE Commun. Surveys Tuts.}, vol.~15, no.~3, pp.
  996--1019, 2013.

\bibitem{Shannon}
C.~E. Shannon, ``Coding theorems for a discrete source with a fidelity
  criterion,'' \emph{IRE Nat. Conv. Rec}, vol.~4, no. 142-163, p.~1, 1959.

\bibitem{Relaying}
Y.~Hu, J.~Gross, and A.~Schmeink, ``On the capacity of relaying with finite
  blocklength,'' \emph{IEEE Trans. Veh. Technol.}, vol.~65, no.~3, pp.
  1790--1794, 2015.

\bibitem{Noma1}
\BIBentryALTinterwordspacing
E.~Dosti, M.~Shehab, H.~Alves, and M.~Latva-aho, ``On the performance of
  non-orthogonal multiple access in the finite blocklength regime,'' \emph{Ad
  Hoc Networks}, vol.~84, pp. 148 -- 157, 2019. [Online]. Available:
  \url{http://www.sciencedirect.com/science/article/pii/S157087051830708X}
\BIBentrySTDinterwordspacing

\bibitem{RS}
Y.~Xu, Y.~Mao, O.~Dizdar, and B.~Clerckx, ``Rate-splitting multiple access with
  finite blocklength for short-packet and low-latency downlink
  communications,'' \emph{arXiv preprint arXiv:2105.06198}, 2021.

\bibitem{Mu_MIMO}
J.~Feng, H.~Q. Ngo, and M.~Matthaiou, ``Coherent {MU-MIMO} in block fading
  channels: A finite blocklength analysis,'' in \emph{IEEE Int. Conf. Commun.
  Workshops (ICC Workshops)}, 2020, pp. 1--6.

\bibitem{URLLC_1}
J.~Östman, A.~Lancho, G.~Durisi, and L.~Sanguinetti, ``{URLLC} with massive
  {MIMO}: Analysis and design at finite blocklength,'' \emph{IEEE Trans.
  Wireless Commun.}, vol.~20, no.~10, pp. 6387--6401, 2021.

\bibitem{URLLC_2}
A.~Lancho, G.~Durisi, and L.~Sanguinetti, ``Cell-free massive {MIMO} with short
  packets,'' in \emph{IEEE Int. Workshop Signal Process. Adv. Wireless Commun.
  (SPAWC)}, 2021, pp. 416--420.

\bibitem{URLLC_4}
H.~Ren, C.~Pan, Y.~Deng, M.~Elkashlan, and A.~Nallanathan, ``Joint power and
  blocklength optimization for {URLLC} in a factory automation scenario,''
  \emph{IEEE Trans. Wireless Commun.}, vol.~19, no.~3, pp. 1786--1801, 2020.

\bibitem{URLLC_5}
W.~J. Ryu and S.~Y. Shin, ``Power allocation for {URLLC} using finite
  blocklength regime in downlink {NOMA} systems,'' in \emph{Int. Conf. Inf.
  Commun. Technol. Convergence (ICTC)}, 2019, pp. 770--773.

\bibitem{URLLC_7}
C.~Shen, T.-H. Chang, H.~Xu, and Y.~Zhao, ``Joint uplink and downlink
  transmission design for {URLLC} using finite blocklength codes,'' in
  \emph{Int. Symp. Wireless Commun. Sys. (ISWCS)}, 2018, pp. 1--5.

\bibitem{URLLC_8}
Z.~Wang, T.~Lv, Z.~Lin, J.~Zeng, and P.~T. Mathiopoulos, ``Outage performance
  of {URLLC} {NOMA} systems with wireless power transfer,'' \emph{IEEE Wireless
  Commun. Lett.}, vol.~9, no.~3, pp. 380--384, 2020.

\bibitem{URLLC_9}
H.~Ren, C.~Pan, Y.~Deng, M.~Elkashlan, and A.~Nallanathan, ``Resource
  allocation for secure {URLLC} in mission-critical {IoT} scenarios,''
  \emph{IEEE Trans. Commun.}, vol.~68, no.~9, pp. 5793--5807, 2020.

\bibitem{multilevel_polar}
M.~{Seidl}, A.~{Schenk}, C.~{Stierstorfer}, and J.~B. {Huber}, ``Polar-coded
  modulation,'' \emph{IEEE Trans. Commun.}, vol.~61, no.~10, pp. 4108--4119,
  2013.

\bibitem{multilevel_polar2}
H.~{Khoshnevis}, I.~{Marsland}, and H.~{Yanikomeroglu}, ``Throughput-based
  design for polar-coded modulation,'' \emph{IEEE Trans. Commun.}, vol.~67,
  no.~3, pp. 1770--1782, 2019.

\bibitem{5G_MLPCM_2}
T.~Koike-Akino, Y.~Wang, D.~S. Millar, K.~Kojima, and K.~Parsons, ``Polar
  coding for multilevel shaped constellations,'' in \emph{{OSA} Signal Process.
  Photonic Commun.}, 2018, p. SpW1G.1.

\bibitem{5G_MLPCM_3}
H.~Khoshnevis, ``Multilevel polar coded-modulation for wireless
  communications,'' Ph.D. dissertation, Carleton University, 2018.

\bibitem{3gpp1}
3GPP, ``{Technical Specification Group Radio Access Network; NR; Multiplexing
  and channel coding},'' {3rd Generation Partnership Project (3GPP)}, TS
  38.212, Release 16, {V}16.8.0, Jan. 2022.

\bibitem{future_MLPCM1}
J.~Dai, D.~Zhang, K.~Niu, Z.~Si, P.~Zhang, S.~Wang, Y.~Yuan \emph{et~al.},
  ``Generalized polarization transform: A novel coded transmission paradigm,''
  \emph{arXiv preprint arXiv:2110.12224}, 2021.

\bibitem{future_MLPCM2}
B.~Tomasi, F.~Gabry, V.~Bioglio, I.~Land, and J.-C. Belfiore, ``Low-complexity
  receiver for multi-level polar coded modulation in non-orthogonal multiple
  access,'' in \emph{2017 IEEE Wireless Communications and Networking
  Conference Workshops (WCNCW)}, 2017, pp. 1--6.

\bibitem{BlockFading}
E.~Malkamaki and H.~Leib, ``Performance of truncated type-{II} hybrid {ARQ}
  schemes with noisy feedback over block fading channels,'' \emph{IEEE Trans.
  Commun.}, vol.~48, no.~9, pp. 1477--1487, 2000.

\bibitem{EID}
M.~D. {Renzo} and W.~{Lu}, ``The equivalent-in-distribution ({EiD})-based
  approach: On the analysis of cellular networks using stochastic geometry,''
  \emph{IEEE Commun. Lett.}, vol.~18, no.~5, pp. 761--764, 2014.

\bibitem{eid2}
P.~C. Pinto and M.~Z. Win, ``Communication in a {Poisson} field of
  interferers--part {I}: Interference distribution and error probability,''
  \emph{IEEE Transactions on Wireless Communications}, vol.~9, no.~7, pp.
  2176--2186, 2010.

\bibitem{r_o}
J.~G. {Andrews}, F.~{Baccelli}, and R.~K. {Ganti}, ``A tractable approach to
  coverage and rate in cellular networks,'' \emph{IEEE Trans. Commun.},
  vol.~59, no.~11, pp. 3122--3134, 2011.

\bibitem{modulation}
C.~{Zhag}, X.~{Mu}, J.~{Yuan}, Y.~{Zhou}, and J.~{Shi}, ``The performance of
  coded modulation with {Gallager} mapping in the finite length regime,'' in
  \emph{Int. Conf. Wireless Commun. Signal Process. (WCSP)}, 2019, pp. 1--6.

\bibitem{Gamma_app_1}
R.~K. Ganti and M.~Haenggi, ``Interference in ad hoc networks with general
  motion-invariant node distributions,'' in \emph{IEEE Int. Symp. Inf.
  Theory}.\hskip 1em plus 0.5em minus 0.4em\relax IEEE, 2008, pp. 1--5.

\bibitem{Gamma_app_2}
M.~Kountouris and N.~Pappas, ``Approximating the interference distribution in
  large wireless networks,'' in \emph{Int. Symp. Wireless Commun. Sys.}\hskip
  1em plus 0.5em minus 0.4em\relax IEEE, 2014, pp. 80--84.

\bibitem{SCD}
H.~{Vangala}, E.~{Viterbo}, and Y.~{Hong}, ``Permuted successive cancellation
  decoder for polar codes,'' in \emph{Int. Symp. Inf. Theory Appl.}, 2014, pp.
  438--442.

\bibitem{Andrew}
J.~G. Andrews, A.~K. Gupta, and H.~S. Dhillon, ``A primer on cellular network
  analysis using stochastic geometry,'' \emph{arXiv preprint arXiv:1604.03183},
  2016.

\bibitem{eta}
F.~Palacio, R.~Agustí, J.~Pérez-Romero, M.~López-Benítez, S.~Grimoud,
  B.~Sayrac, I.~Dagres, A.~Polydoros, J.~Riihijärvi, J.~Nasreddine,
  P.~Mähönen, L.~Gavrilovska, V.~Atanasovski, and J.~Beek, ``Radio
  environmental maps: information models and reference model. document number
  {D4.1},'' \emph{Deliverable D4.1 del projecte Europeu FARAMIR (248351)},
  2011.

\bibitem{Meta1}
M.~Haenggi, ``The meta distribution of the {SIR} in {Poisson} bipolar and
  cellular networks,'' \emph{IEEE Trans. Wireless Commun.}, vol.~15, no.~4, pp.
  2577--2589, 2016.

\bibitem{reducedPalm_book}
------, \emph{Stochastic Geometry for Wireless Networks}.\hskip 1em plus 0.5em
  minus 0.4em\relax Cambridge University Press, 2012.

\bibitem{Meta2}
H.~ElSawy and M.-S. Alouini, ``On the meta distribution of coverage probability
  in uplink cellular networks,'' \emph{IEEE Commun. Lett.}, vol.~21, no.~7, pp.
  1625--1628, 2017.

\bibitem{Meta3}
H.~Ibrahim, H.~Tabassum, and U.~T. Nguyen, ``The meta distributions of the
  {SIR/SNR} and data rate in coexisting sub-6{GHz} and millimeter-wave cellular
  networks,'' \emph{IEEE Open J. Commun. Soc.}, vol.~1, pp. 1213--1229, 2020.

\bibitem{Meta_noise}
W.~Hassan, H.-S. Jo, S.~Ikki, and M.~Nekovee, ``Spectrum-sharing method for
  co-existence between {5G} {OFDM}-based system and fixed service,'' \emph{IEEE
  Access}, vol.~7, pp. 77\,460--77\,475, 01 2019.

\bibitem{Noma4}
K.~S. Ali, M.~Haenggi, H.~ElSawy, A.~Chaaban, and M.-S. Alouini, ``Downlink
  non-orthogonal multiple access ({NOMA}) in {Poisson} networks,'' \emph{IEEE
  Trans. Commun.}, vol.~67, no.~2, pp. 1613--1628, 2018.

\bibitem{fubini}
G.~Fubini, ``Sugli integrali multipli,'' \emph{Rend. Acc. Naz. Lincei},
  vol.~16, pp. 608--614, 1907.

\end{thebibliography}
